%% file: Radermacher_Bianchistifffluid.tex
\documentclass[final]{amsart}

\include{header_paperenglish}

\begin{document}

\title[Bianchi~B stiff fluids initial singularity]{Orthogonal Bianchi~B stiff fluids close to the initial singularity}
\author[K. Radermacher]{Katharina Radermacher}
\address{Department of Mathematics, KTH Royal Institute of Technology, SE-10044 Stockholm, Sweden}
\curraddr{}
\email{kmra@kth.se}
\urladdr{}
\dedicatory{}
\date{\today}
\translator{}

\renewcommand{\S}{{\mathcal S}}
\newcommand{\slimit}{\operatorname s}
\newcommand{\rfactorstiff}{\tilde{\operatorname r}}
\newcommand{\integraltwominusq}{\R}
\newcommand{\betaN}{\beta_{N_+}}
\newcommand{\betaD}{\beta_\Delta}
\newcommand{\stiffalphalimit}{\J}
\newcommand{\stiffalphalimitminus}{\stiffalphalimit_{-}}
\newcommand{\stiffalphalimitzero}{\stiffalphalimit_{\, 0}}
\newcommand{\stiffalphalimitplus}{\stiffalphalimit_+}
\newcommand{\taubone}{\operatorname{T1}}
\newcommand{\taubtwo}{\operatorname{T2}}
\newcommand{\bparamk}{\kappa}
\newcommand{\binvparam}{\eta}
\newcommand{\functiondecaylemma}{\zeta}

\newcommand{\stildelimit}{\tilde{\operatorname{s}}}
\newcommand{\centreunstable}{M}

\newcommand{\considerconvtoD}{Consider a solution to \equs ~\eqref{eqns_evolutionbianchib}--\eqref{eqn_evolutionomega} converging to $(\slimit,\stildelimit,0,0,0)\in\stiffalphalimit$ as~$\tau\rightarrow{-}\infty$}
\newcommand{\considerconvtoaxisD}{Consider a solution to \equs ~\eqref{eqns_evolutionbianchib}--\eqref{eqn_evolutionomega} converging to $(\slimit,0,0,0,0)\in\stiffalphalimit$ as~$\tau\rightarrow{-}\infty$}
\newcommand{\considerconvoffaxisD}{Consider a solution to \equs ~\eqref{eqns_evolutionbianchib}--\eqref{eqn_evolutionomega} converging to $(\slimit,\stildelimit,0,0,0)\in\stiffalphalimit$ as~$\tau\rightarrow{-}\infty$, where~$\stildelimit>0$}

\begin{abstract}
	In our previous article~\cite{radermacher_sccogonbianchibperfectfluidsvacuum}, we investigated the asymptotic behaviour of orthogonal Bianchi class B perfect fluids close to the initial singularity and proved the Strong Cosmic Censorship conjecture in this setting. In several of the statements, the case of a stiff fluid had to be excluded. The present paper fills this gap.
	
	We work in expansion-normalised variables introduced by Hewitt--Wainwright and find that solutions converge, but show a convergence behaviour very different from the non-stiff case: All solutions tend to limit points in the two-dimensional Jacobs set. A set of full measure, which is also a countable intersection of open and dense sets in the state space, yields convergence to a specific subset of the Jacobs set.
\end{abstract}

\maketitle

\section{Introduction}

When modelling idealised distributions of matter in General Relativity, for example in a homogeneous model of the universe, a common choice of matter model is that of a perfect fluid, which includes dust and radiation as specific parameter choices.
The extremal parameter choice of a stiff fluid in particular is sometimes used to model the universe at very early times.
In the present paper, we consider stiff fluids which additionally exhibit Bianchi symmetry.

Besides this exterior motive to investigate stiff fluids, the present paper serves an additional purpose: In our earlier paper~\cite{radermacher_sccogonbianchibperfectfluidsvacuum}, we have investigated the Strong Cosmic Censorship conjecture in certain Bianchi perfect fluid models. While our affirmative answer to the conjecture extends to the case of a stiff fluid, many of our additional results on the asymptotic behaviour towards the initial singularity do not. The present paper fills this gap by treating in detail the asymptotic behaviour of stiff fluid \ogon\ Bianchi class~B solutions.

\subsection{Bianchi stiff fluid spacetimes}

We consider four-dimensional spacetimes~$(M,g)$ which satisfy Einstein's field equations
\begin{equation}
	R_{\alpha \beta} -\frac12 S g_{\alpha \beta} = T_{\alpha \beta}
\end{equation}
for the stress-energy tensor of a perfect fluid
\begin{equation}
\label{eqn_perfectfluid}
	T_{\alpha \beta} = \mu u_\alpha u_\beta +p (g_{\alpha \beta} + u_\alpha u_\beta).
\end{equation}
We assume that the pressure~$p$ and the energy density~$\mu$ are related by a linear equation of state~$p=(\gamma-1)\mu$ with a parameter~$\gamma\in[0,2]$. By setting this parameter to the extremal value~$\gamma=2$, in other words
\begin{equation}
	p=\mu,
\end{equation}
we restrict the discussion to the case of a stiff fluid.

In terms of symmetry, we assume that the spacetimes under consideration arise from initial data with a left-invariant metric and \fundform, meaning that they are invariant under the action of a three-dimensional \liegr~$G$. 
The three-dimensional \liegr~$G$ then acts on the \mghd~$(M,g)$,
\begin{equation}
	G\times M\rightarrow M,
\end{equation}
and leaves the four-dimensional metric~$g$ invariant.
We further assume the unit timelike vector field~$u$ appearing in~\eqref{eqn_perfectfluid} to be \ogon\ to the orbits of the group action, which defines non-tilted perfect fluids.

Three-dimensional \liegr s have been investigated and classified by Bianchi in the 1900s, and the corresponding spacetimes have inherited the name. For a historical account on the classification of the groups, we refer the reader to~\cite{krasinskibehrschueckingestabrookwahlquistellisjantzenkundt_bianchiclass}. Lie groups are called of class~A or~B depending on whether they are unimodular or not. Each class is further separated into several types:
\begin{itemize}
	\item Unimodular, \ie of Bianchi class~A: types~I, II, VI$_0$, VII$_0$, VIII, IX;
	\item Non-unimodular, \ie of Bianchi class~B: types~IV, V, VI$_{\binvparam}$, VII$_{\binvparam}$, with a parameter~$\binvparam\in\RR\setminus\{0\}$.
\end{itemize}
For the current paper, we are interested in \liegr s of Bianchi class~B. Types~I and~II appear as boundary cases.

The setting of \ogon\ Bianchi class~B perfect fluid spacetimes, but with a general parameter~$\gamma$, is what we have investigated in~\cite{radermacher_sccogonbianchibperfectfluidsvacuum}. We have shown there, in Section~11, that we find spacetimes solving Einstein's equations with the correct stress-energy tensor and having the requested symmetry if we start with the \liegr~$G$ as initial data \mf, equipped with a metric and second fundamental form invariant under its own action, together with a constant~$\mu_0$ such that Einstein's constraint equations are satisfied. The \mghd\ then has the form
\begin{equation}
\label{eqn_mghdBianchi}
	(I\times G,\, g={-}dt^2 + {}^tg),
\end{equation}
with~$I$ an interval and~$\{{}^tg\}_{t\in I}$ a family of \riem\ metrics on~$G$ which are invariant under the action of~$G$.
\begin{rema}
\label{rema_excludeexceptionalBianchi}
	Certain spacetimes where the \liegr\ is of Bianchi type~VI$_{{-}1/9}$ are called 'exceptional', as Einstein's constraint equations in these cases allow for an additional degree of freedom in the initial data, see~\cite[Lemma~11.13]{radermacher_sccogonbianchibperfectfluidsvacuum}. As the evolution equations we use in the following do not apply in the 'exceptional' cases, we restrict our discussion to the '\ogon ' setting, which denotes the remaining, 'non-exceptional' cases. We refer to~\cite[Rem.~11.14]{radermacher_sccogonbianchibperfectfluidsvacuum} for more details on this naming.
\end{rema}
We have found in~\cite[Prop.~11.24]{radermacher_sccogonbianchibperfectfluidsvacuum} that for Bianchi class~B spacetimes as in~\eqref{eqn_mghdBianchi} other than Minkowski or \desitter, the lower bound of the maximal interval of existence~$I=(t_-,t_+)$ is finite, $t_->{-}\infty$.
The focus of interest in this paper is to investigate the properties of timeslices~$(\{t\}\times G,{}^tg)$ as~$t$ tends towards the past and approaches this initial singularity.

\subsection{Expansion-normalised evolution equations}

In the setting described in the previous subsection, finding the \mghd\ for given initial data is equivalent to finding the solution curve through a given point for a specific ordinary differential equation with constraints in~$\RR^5$. The five-dimensional variables and evolution equations have been introduced in~\cite{hewittwainwright_dynamicalsystemsapproachbianchiorthogonalB}, and they result from an algebraic manipulation of quantities defined on each timeslice~$\{t\}\times G$, followed by a normalisation with the mean \curv\ of the timeslice. They are called expansion-normalised variables, and their counterpart in Bianchi class~A is known as Wainwright-Hsu variables and was first introduced in~\cite{wainwrighthsu_dynamicalsystemsapproachbianchiorthogonalA}. Equivalence between the geometric setting of finding the \mghd\ and the dynamical setting of solving the five-dimensional ODE has been proven in~\cite[Sect.~11]{radermacher_sccogonbianchibperfectfluidsvacuum}.

In our situation of a stress-energy tensor for a stiff fluid, the dynamical setting is the following:
For a fixed parameter~$\bparamk\in\RR$, the evolution equations for the variables~$(\Sigma_+,\tilde\Sigma,\Delta,\tilde A,N_+)$ are
\begin{equation}
\label{eqns_evolutionbianchib}
\begin{subaligned}
	\Sigma_+'={}&(q-2)\Sigma_+-2\tilde N,\\
	\tilde\Sigma'={}&2(q-2)\tilde\Sigma-4\Sigma_+\tilde A-4\Delta N_+,\\
	\Delta'={}&2(q+\Sigma_+-1)\Delta+2(\tilde\Sigma-\tilde N)N_+,\\
	\tilde A'={}&2(q+2\Sigma_+)\tilde A,\\
	N_+'={}&(q+2\Sigma_+)N_++6\Delta.
\end{subaligned}
\end{equation}
In these equations, one denotes by~$'$ differentiation ${d}/{d\tau}$ with respect to the newly introduced time~$\tau$ which is related to the original time~$t$ via~$dt/d\tau=3/\theta$, 
where~$\theta$ is the mean curvature of the timeslice~$(\{t\}\times G,{}^tg)$.
In these evolution equations, the deceleration parameter~$q$ satisfies
\begin{equation}
\label{eqn_qwithaandn}
	q=2-2\tilde A-2\tilde N,
\end{equation}
and one defines
\begin{equation}
\label{eqn_definitiontilden}
	\tilde N=\frac13(N_+^2-{\bparamk}\tilde A).
\end{equation}
The evolution is constrained to the subset of~$\RR^5$ which satisfies
\begin{equation}
\label{eqn_constraintgeneralone}
	\tilde\Sigma\tilde N-\Delta^2-\Sigma_+^2\tilde A=0,
\end{equation}
as well as
\begin{equation}
\label{eqn_constraintgeneraltwo}
	 \tilde\Sigma\ge0, \qquad \tilde A\ge0, \qquad\tilde N\ge0, \qquad \Sigma_+^2+\tilde\Sigma+\tilde A+\tilde N< 1.
\end{equation}
Equations~\eqref{eqn_constraintgeneralone} and~\eqref{eqn_constraintgeneraltwo} are preserved by the evolution equations~\eqref{eqns_evolutionbianchib}, see~\cite[Remark~3.2]{radermacher_sccogonbianchibperfectfluidsvacuum}.
The density parameter defined by
\begin{equation}
\label{eqn_omegageneral}
	\Omega=1-\Sigma_+^2-\tilde\Sigma-\tilde A-\tilde N
\end{equation}
satisfies~$\Omega>0$ due to the constraints~\eqref{eqn_constraintgeneraltwo} and, as a consequence of the differential equations~\eqref{eqns_evolutionbianchib}, evolves according to
\begin{equation}
\label{eqn_evolutionomega}
	\Omega'=(2q-4)\Omega.
\end{equation}
This setting is a special case of the expansion-normalised variables for perfect fluid \ogon\ Bianchi class~B solutions which we recall in Appendix~\ref{appendix_nonstiffevolution}.

The sign of the parameter~$\bparamk$ is used to distinguish between the different Bianchi types of class~B:
\begin{itemize}
	\item $\bparamk<0$: Bianchi type~VI$_{\binvparam}$ with~$\binvparam=1/\bparamk$,
	\item $\bparamk>0$: Bianchi type~VII$_{\binvparam}$ with~$\binvparam=1/\bparamk$,
	\item $\bparamk=0$: Bianchi type~IV or type V.
\end{itemize}
More precisely, the different Bianchi types of class~B occupy different subsets of the set defined by equations~\eqref{eqn_constraintgeneralone}--\eqref{eqn_constraintgeneraltwo}, and these subsets are invariant under the evolution equations~\eqref{eqns_evolutionbianchib}. For the decomposition of the state space into different parts, we refer the reader to~\cite{radermacher_sccogonbianchibperfectfluidsvacuum}. Table~1 in this reference lists the subsets corresponding to the four Bianchi types of class~B. The same subsets also appear in the statement of Theorem~\ref{theo_fullmeasurelimitset} further down in the present paper. Two Bianchi types of class~A, namely type~I and~II, are contained in the state space as subsets in which~$\tilde A=0$ holds, and they are given in Table~2 of~\cite{radermacher_sccogonbianchibperfectfluidsvacuum}. In these Bianchi types, the parameter~$\bparamk$ is of no importance. Several solutions to the evolution equations having additional symmetry are listed in Table~3 of~\cite{radermacher_sccogonbianchibperfectfluidsvacuum}.

In the present paper, we are interested in the behaviour of solution curves as~$\tau\rightarrow{-}\infty$, which according to~\cite[Sect.~11]{radermacher_sccogonbianchibperfectfluidsvacuum} is equivalent to the limit~$t\rightarrow t_-$ in the \mghd.

\subsection{Previous results}

The evolution equations with constraints, equations~\eqref{eqns_evolutionbianchib}--\eqref{eqn_evolutionomega}, have been introduced in~\cite{hewittwainwright_dynamicalsystemsapproachbianchiorthogonalB} in the general case, encompassing all orthogonal Bianchi class~B perfect fluid solutions, not only the stiff fluid case. In particular, in this reference a number of equilibrium sets are identified and investigated, where an equilibrium set denotes a set of points satisfying the constraint equations~\eqref{eqn_constraintgeneralone}--\eqref{eqn_constraintgeneraltwo} such that the \rhs\ of the evolution equations~\eqref{eqns_evolutionbianchib} is zero. For our discussion of stiff fluid solutions, the Jacobs set~$\stiffalphalimit$ is of central importance, as it governs the asymptotic behaviour as~$\tau\rightarrow{-}\infty$.
\begin{defi}
\label{defi_setD}
The Jacobs set~$\stiffalphalimit$ is defined by
\begin{equation}
    \Sigma_+^2+\tilde\Sigma<1,\qquad \Delta=\tilde A=N_+=0.
\end{equation}
\end{defi}
We remark that in~\cite{hewittwainwright_dynamicalsystemsapproachbianchiorthogonalB} the Jacobs set was denoted by~$\D$, see also Table~4 in that reference.

One easily checks that all points in~$\stiffalphalimit$ satisfy the constraint equations~\eqref{eqn_constraintgeneralone}--\eqref{eqn_constraintgeneraltwo} and that the \rhs\ of the evolution equations~\eqref{eqns_evolutionbianchib} vanishes. Further, the Jacobs set~$\stiffalphalimit$ is characterised by~$q=2$, which is easily deduced from the definition of~$q$, equation~\eqref{eqn_qwithaandn}, and the constraint equation~\eqref{eqn_constraintgeneralone}.

\begin{figure}[ht]
\begin{tikzpicture}[xscale=2.5,yscale=2.5]
	\begin{scope}
		\clip[domain=-1:1] plot (\x, {(\x+1)*(1-\x)});
		\fill [fill=gray!50!white] (-1.2,0) -- (1.2,0) -- (1.2,1) -- (-1.2,1);
	\end{scope}
	\draw [->] (-1.2,0) -- (1.2,0);
	\draw [->] (0,0) -- (0,1.2);
	\draw[gray, domain=-1:1] plot (\x, {(\x+1)*(1-\x)});
	\draw[fill] (-1,0) circle [radius=0.025];
	\node [below] at (0,0) {$0$};
	\draw (1,-0.03) -- (1,0.04);
	\node [below] at (1,0) {$1$};
	\node [below] at (-1,0) {${-}1$};
	\draw (-0.04,1) -- (0.04,1);
	\node [above left] at (0,1) {$1$};
	\node [right] at (1.2,0) {$\Sigma_+$};
	\node [above] at (0,1.2) {$\tilde\Sigma$};
	\node [above right] at (0.2,0.3) {\textcolor{black}{$\stiffalphalimit$}};
	\end{tikzpicture}
	\caption{The Jacobs set~$\stiffalphalimit$, projected to the $\Sigma_+\tilde\Sigma$-plane.}
	\label{figure_setJ}
\end{figure}

When considering the linearisation of the evolution equations~\eqref{eqns_evolutionbianchib} without constraining them to equations~\eqref{eqn_constraintgeneralone}--\eqref{eqn_constraintgeneraltwo}, one can draw conclusions about the local stability. The explicit form of this vector field is given in Appendix~\ref{appendix_linearisedevolution}, its eigenvalues are\footnote{Note that there appears to be a typo in~\cite{hewittwainwright_dynamicalsystemsapproachbianchiorthogonalB} where these eigenvalues are stated. The same typo also occurs for the eigenvalues on the Kasner parabola. In the latter case, we have given a corrected version in \cite[App~A]{radermacher_sccogonbianchibperfectfluidsvacuum}. In both cases, the typos do however not affect the signs of the eigenvalues.}
\begin{equation}
	0 \qquad 2(1+\Sigma_+\pm\sqrt{3\tilde\Sigma}) \qquad 4(1+\Sigma_+) \qquad 0.
\end{equation}
The double occurence of the eigenvalue~$0$ reflects the fact that~$\stiffalphalimit$ is a set of equilibrium points of dimension~two.
Two of the non-zero eigenvalues are strictly positive in~$\stiffalphalimit$, while the remaining one, the eigenvalue~$2(1+\Sigma_+-\sqrt{3\tilde\Sigma})$, can have either sign. We divide the Jacobs set~$\stiffalphalimit$ into three subsets according to the sign of this eigenvalue.
\begin{defi}
	The subsets~$\stiffalphalimitminus$, $\stiffalphalimitzero$, and~$\stiffalphalimitplus$ of the Jacobs set~$\stiffalphalimit$ are defined by
	\begin{align}
		\stiffalphalimitminus\coloneqq{}& \stiffalphalimit \cap \{(1+\Sigma_+)^2<3\tilde\Sigma\},\\
		\stiffalphalimitzero\coloneqq{}& \stiffalphalimit \cap \{(1+\Sigma_+)^2=3\tilde\Sigma\},\\
		\stiffalphalimitplus\coloneqq{}& \stiffalphalimit \cap \{(1+\Sigma_+)^2>3\tilde\Sigma\}.
	\end{align}
\end{defi}
\begin{figure}[ht]
\begin{tikzpicture}[xscale=2.5,yscale=2.5]
	\begin{scope}
		\clip[domain=-1:1] plot (\x, {(\x+1)*(1-\x)});
		\clip[domain=-1:0.5] plot (\x, {(1+\x)^2/3}) -- (1.1,1) -- (-1.1,1);
		\fill [fill=green!50!black!60!white] (-1.2,0) -- (1.2,0) -- (1.2,1) -- (-1.2,1);
	\end{scope}
	\begin{scope}
		\clip[domain=-1:1] plot (\x, {(\x+1)*(1-\x)});
		\clip[domain=-1:0.5] plot (\x, {(1+\x)^2/3}) -- (1.1,1) -- (1.1,0);
		\fill [fill=orange!60!white] (-1.2,0) -- (1.2,0) -- (1.2,1) -- (-1.2,1);
	\end{scope}
	\draw [->] (-1.2,0) -- (1.2,0);
	\draw [->] (0,0) -- (0,1.2);
	\draw[gray, , domain=-1:1,] plot (\x, {(\x+1)*(1-\x)});
	\draw[black, thick, domain=-1:0.5,] plot (\x, {(1+\x)^2/3});
	\draw[fill] (-1,0) circle [radius=0.025];
	\node [below] at (0,0) {$0$};
	\draw (1,-0.03) -- (1,0.04);
	\node [below] at (1,0) {$1$};
	\node [below] at (-1,0) {${-}1$};
	\draw (-0.04,1) -- (0.04,1);
	\node [above left] at (0,1) {$1$};
	\node [right] at (1.2,0) {$\Sigma_+$};
	\node [above] at (0,1.2) {$\tilde\Sigma$};
	\node [above left] at (-1,0) {$\taubone$};
		\draw [dashed] (0.5,0) -- (0.5,0.75);
		\node [below] at (0.5,0) {$1/2$};
	\draw[fill] (0.5,0.75) circle [radius=0.025];
	\node [above right] at (0.5,0.75) {$\taubtwo$};
	\node [above left] at (-0.25, 0.15) {$\stiffalphalimitzero$};
	\node [above left] at (-0.45, 0.7) {\textcolor{green!50!black}{$\stiffalphalimitminus$}};
	\node [above right] at (0.77, 0.3) {\textcolor{orange}{$\stiffalphalimitplus$}};
	\end{tikzpicture}
	\caption{Inside the Jacobs set~$\stiffalphalimit$, the arc~$\stiffalphalimitzero$ joins the two Taub points~$\taubone$ and~$\taubtwo$. The subsets~$\stiffalphalimitminus$ (green) and~$\stiffalphalimitplus$ (orange) are situated above and below of~$\stiffalphalimitzero$, respectively.}
	\label{figure_subsetsD}
\end{figure}
Using the signs of the eigenvalues, Hewitt--Wainwright in~\cite{hewittwainwright_dynamicalsystemsapproachbianchiorthogonalB} identify points in~$\stiffalphalimitplus$ as local sources and points in~$\stiffalphalimitminus$ as saddles. On~$\stiffalphalimitzero$, three eigenvalues vanish, and no statement about the local stability can be given this way.
Applying dynamical system methods, they further show that all~$\alpha$-limit sets of solutions to~\eqref{eqns_evolutionbianchib}--\eqref{eqn_evolutionomega} are contained in the Jacobs set~$\stiffalphalimit$. We give a refined proof of this last statement in Proposition~\ref{prop_alphalimitsets_stiff}.

\subsection{New results: convergence behaviour at the initial singularity}

The objective of the present paper is to investigate in more detail the asymptotic behaviour of solutions to the evolution equations~\eqref{eqns_evolutionbianchib}--\eqref{eqn_evolutionomega} as~$\tau\rightarrow{-}\infty$. After having determined the~$\alpha$-limit set in Proposition~\ref{prop_alphalimitsets_stiff}, we strengthen this statement in Proposition~\ref{prop_convergencetoD}: For every solution, there is exactly one~$\alpha$-limit point~$(\slimit,\stildelimit,0,0,0)\in\stiffalphalimit$ to which the solution converges as~$\tau\rightarrow{-}\infty$.

The different parts of the Jacobs set~$\stiffalphalimit$ have very different qualitative properties regarding which and how many solutions converge to a limit point contained in them. We identify the subset~$\stiffalphalimitminus$ as unstable in the sense that only very specific solutions and limit points can occur: All solutions which are not of Bianchi class~A have limit points situated on a specific arc in~$\stiffalphalimitminus$, and the solutions are contained in a~$C^1$ sub\mf\ of positive codimension, see Theorem~\ref{theo_unstablepartfamilymanifolds}. 
A similar statement holds for convergence towards the special arc~$\stiffalphalimitzero$: Solutions of Bianchi class~B can only converge to one of at most two points on the arc~$\stiffalphalimitzero$, see Proposition~\ref{prop_specialarc} and Remark~\ref{rema_atmosttwopointsstiffzero}. The exact location of the possible Bianchi class~B limit points depends on the choice of parameter~$\bparamk$. 
This restriction to specific limit points 
is captured in the following proposition which applies to convergence to~$\stiffalphalimitminus$ and~$\stiffalphalimitzero$ and is proven in Proposition~\ref{prop_convergencemain_unstablepart}.
\begin{prop}
	Consider a solution to equations~\eqref{eqns_evolutionbianchib}--\eqref{eqn_evolutionomega} converging to~$(\slimit,\stildelimit,0,0,0)\in\stiffalphalimit$ as~$\tau\rightarrow{-}\infty$. If~$(1+\slimit)^2\le3\stildelimit$, then
	\begin{equation}
		\tilde A(3\slimit^2+{\bparamk}\stildelimit)=0
	\end{equation}
	holds along the whole solution.
\end{prop}
We find that convergence to a limit point on the~$\Sigma_+$-axis imposes a restriction in a similar fashion: If the solution is of Bianchi class~B, then the only possible limit point on the axis is the origin, see Lemma~\ref{lemm_decaystiff_axis} and Remark~\ref{rema_limitpointsaxis}.
\begin{prop}
	Consider a solution to equations~\eqref{eqns_evolutionbianchib}--\eqref{eqn_evolutionomega} converging to~$(\slimit,0,0,0,0)\in\stiffalphalimit$ as~$\tau\rightarrow{-}\infty$. Then
	\begin{equation}
		\tilde A\slimit=0
	\end{equation}
	holds along the whole solution.
\end{prop}

Considering all possible solutions to equations~\eqref{eqns_evolutionbianchib}--\eqref{eqn_evolutionomega} as a whole, we find that almost all converge to~$\stiffalphalimitplus$ and possibly some additional special points. 'Almost all' in this setting is well-defined both in a measure-theoretic and a topological way. This is the statement of the following theorem, which constitutes the main result of the present paper. The proof is given in Section~\ref{section_proofmaintheorem}.
\begin{theo}
\label{theo_fullmeasurelimitset}
	The following holds for solutions to the evolution \equs~\eqref{eqns_evolutionbianchib}--\eqref{eqn_evolutionomega}:
	\begin{itemize}
		\item Consider the set describing Bianchi type~VI$_{\binvparam}$, \ie the set
		\begin{equation}
			B(VI_{\binvparam}) ={}\{\eqref{eqn_constraintgeneralone}-\eqref{eqn_constraintgeneraltwo}\text{ hold},\,{\bparamk}=\nicefrac1{\binvparam}<0,\, \tilde A>0\}.
		\end{equation}
		Then the subset of points such that the corresponding solution converges to a point in~$\stiffalphalimitplus$ or to a point in~$\stiffalphalimitzero\cap\{3\Sigma_+^2+\bparamk\tilde\Sigma=0\}$, as~$\tau\rightarrow{-}\infty$, is
		of full measure and a countable intersection of open and dense sets
		in B(VI$_{\binvparam}$).
		\item Consider the set describing Bianchi type~VII$_{\binvparam}$, \ie the set
		\begin{equation}
			B^\pm(VII_{\binvparam}) ={}\{\eqref{eqn_constraintgeneralone}-\eqref{eqn_constraintgeneraltwo}\text{ hold},\, {\bparamk}=\nicefrac1{\binvparam}>0,\,\tilde A>0,\,N_+>0\textit{ or }N_+<0\}.
		\end{equation}
		Then every solution converges to a point in~$\stiffalphalimitplus$ as~$\tau\rightarrow{-}\infty$.
		\item Consider the set describing Bianchi type~IV, \ie the set
		\begin{equation}
			B^\pm(IV) =\{\eqref{eqn_constraintgeneralone}-\eqref{eqn_constraintgeneraltwo}\text{ hold},\, {\bparamk}=0,\,\tilde A>0,\,N_+>0\textit{ or }N_+<0\}.
		\end{equation}
		Then every solution converges to a point in~$\stiffalphalimitplus$ or to a point in~$\stiffalphalimit\cap\{\Sigma_+=0\}$ as~$\tau\rightarrow{-}\infty$.
		\item Consider the set describing Bianchi type~V, \ie the set
		\begin{equation}
			B(V)=\{\eqref{eqn_constraintgeneralone}-\eqref{eqn_constraintgeneraltwo}\text{ hold},\,
			{\bparamk}=0,\,\tilde A>0,\,\Sigma_+=\Delta=N_+=0\}.
		\end{equation}
		Then every solution converges to a point in~$\stiffalphalimit\cap\{\Sigma_+=0\}$ as~$\tau\rightarrow{-}\infty$.
	\end{itemize}
\end{theo}
\begin{rema}
	Recall that we have excluded 'exceptional' Bianchi type~VI$_{{-}1/9}$ spacetimes, see Remark~\ref{rema_excludeexceptionalBianchi}. As a consequence, our statement in Theorem~\ref{theo_fullmeasurelimitset} does not include all Bianchi type~VI$_\binvparam$ spacetimes, even though it applies to the choice of parameter~$\bparamk={{-}9}$ in the evolution equations~\eqref{eqns_evolutionbianchib}--\eqref{eqn_evolutionomega}.
\end{rema}

\begin{rema}
	For solutions of Bianchi type~IV and~V, the previous theorem allows for convergence to limit points with~$\Sigma_+=0$. In case of Bianchi~V, this is the only option, while in Bianchi~IV this might be a relict of the choice of coordinates, turning the surface defined by the constraint equation~\eqref{eqn_constraintgeneralone} singular in these points.
\end{rema}

\subsection{Comparison to the results in the non-stiff fluid setting}

The results in this paper should be read alongside those of~\cite{radermacher_sccogonbianchibperfectfluidsvacuum}, where we investigated the asymptotic behaviour in the case of \ogon\ Bianchi class~B perfect fluid solution, but had to exclude the case of a stiff fluid. The present paper completes this discussion.

In terms of expansion-normalised variables, the stiff fluid case merely corresponds to setting one parameter, called~$\gamma$, to its extremal value. The full set of evolution equations in the general perfect fluid setting is given as equations~(5)--(11) in~\cite{radermacher_sccogonbianchibperfectfluidsvacuum} and was developed in~\cite{hewittwainwright_dynamicalsystemsapproachbianchiorthogonalB}. In Appendix~\ref{appendix_nonstiffevolution} we give more details on how to obtain the stiff fluid case from the general case.

Even though the evolution equations are very closely related, the asymptotic behaviour in the case of stiff fluids differs significantly from that in non-stiff fluids. As~$\tau\rightarrow{-}\infty$, stiff fluid solutions converge to a limit point in~$\stiffalphalimit$ and~$\Omega$ tends to a positive limit value. In the non-stiff fluid setting, at least for parameters~$\gamma\in[0,2/3)$ where this could be shown, limit points of non-constant solutions are contained in one of two arcs: the Kasner parabola~$\kasnerparabola$ defined by
\begin{equation}
	\Sigma_+^2+\tilde\Sigma=1,\qquad \Delta=\tilde A=N_+=0,
\end{equation}
see~\cite[Def.~1.15]{radermacher_sccogonbianchibperfectfluidsvacuum},
which is contained in the boundary of the Jacobs set~$\stiffalphalimit$, and the plane wave equilibrium equilibrium points. Both arcs are contained in vacuum~$\Omega=0$. For a definition of the plane wave equilibrium set as well as the statement, we refer to Definition~1.17 as well as Proposition~4.2, Proposition~4.4 and Proposition~6.1 in~\cite{radermacher_sccogonbianchibperfectfluidsvacuum}.

Despite the different locations of the limit points, there are striking similarities between the Jacobs set~$\stiffalphalimit$ for stiff fluids and the Kasner parabola~$\kasnerparabola$ for non-stiff fluids. More precisely, the different subset~$\stiffalphalimitminus$, $\stiffalphalimitzero$ and~$\stiffalphalimitplus$ have counterparts in~$\kasnerparabola$ with similar behaviour of the linearisation of the evolution equations, compare~\cite[App.~A.1]{radermacher_sccogonbianchibperfectfluidsvacuum}:
\begin{itemize}
	\item In~$\stiffalphalimitminus$ (stiff) as well as in the subset~$\kasnerparabola\cap\{{-}1<\Sigma_+<1/2\}$ (non-stiff), the number of zero eigenvalues equals the dimension of the equilibrium set. One eigenvalue is negative, all others are positive.
	\item In~$\stiffalphalimitzero$ (stiff) as well as in the point Taub~2, which is the point~$\kasnerparabola\cap\{\Sigma_+=1/2\}$ (non-stiff), the number of zero eigenvalues exceeds the dimension of the equilibrium set by one. All others eigenvalues are positive.
	\item In~$\stiffalphalimitplus$ (stiff) as well as in the subset~$\kasnerparabola\cap\{\Sigma_+>1/2\}$ (non-stiff), the number of zero eigenvalues equals the dimension of the equilibrium set. All other eigenvalues are positive.
\end{itemize}
These similaritites carry through to similarities of the asymptotic behaviour of solutions to the full evolution equations:
\begin{itemize}
	\item A point in~$\stiffalphalimitminus\cup\stiffalphalimitzero$ (stiff) as well as a point in the subset~$\kasnerparabola\cap\{{-}1<\Sigma_+\le1/2\}$ (non-stiff) can be the limit point of solutions of class~B ($\tilde A>0$) only if it satisfies a specific relation which includes the parameter~$\bparamk$. For the statement in the non-stiff fluid case, see~\cite[Prop.~1.19]{radermacher_sccogonbianchibperfectfluidsvacuum}.
	\item The set~$\stiffalphalimitplus$ (stiff) as well as the arc~$\kasnerparabola\cap\{\Sigma_+>1/2\}$  (non-stiff) are the sets of points to which 'most' solutions converge, in a well-defined measure-theoretic and topological sense. In the non-stiff setting, this is the statement of~\cite[Thm~1.18]{radermacher_sccogonbianchibperfectfluidsvacuum}.
	
	Convergence to limit points with~$\Sigma_+=0$ for solutions of Bianchi type~IV might be a relict of the choice of coordinates both in Theorem~\ref{theo_fullmeasurelimitset} in the present paper and in~\cite[Thm~1.18]{radermacher_sccogonbianchibperfectfluidsvacuum}.
\end{itemize}
From a geometric point of view, the different limit sets in stiff and non-stiff fluids have more in common than one might judge from a first glance.
In fact, both the Kasner parabola~$\kasnerparabola$ and the Jacobs set~$\stiffalphalimit$ are contained in the invariant subset describing Bianchi type~I, see also~\cite[Table~2]{radermacher_sccogonbianchibperfectfluidsvacuum}. Neglecting for a moment the existence of (negligibly many, see our comparison above or~\cite[Thm~10.3]{radermacher_sccogonbianchibperfectfluidsvacuum}) non-stiff solutions converging to plane wave equilibrium points, both stiff and non-stiff perfect fluid solutions converge to Bianchi type~I limit points as~$\tau\rightarrow{-}\infty$. In terms of initial data to Einstein's field equations, this implies that the data induced on Cauchy hypersurfaces~$\{t\}\times G$ approaches Bianchi type~I initial data as~$t\rightarrow t_-$.

The techniques we employ in the present paper ressemble that of~\cite{radermacher_sccogonbianchibperfectfluidsvacuum}: The main work is carried out in an extensive discussion of decay and convergence rates of the individual variables, depending in part on the location of the limit point.
Once this is achieved in Section~\ref{section_asymptoticconvergence}, a more detailed investigation of the global behaviour of solutions converging to each of the different subsets ensues, building on the relations between the different decay rates. For some of the subsets, we additionally employ dynamical systems techniques.

We note at this point that the expression for~$q$ in stiff fluids, equation~\eqref{eqn_qwithaandn}, is identical to the form we have available in vacuum models, see~\cite[eq.~(14)]{radermacher_sccogonbianchibperfectfluidsvacuum}.
This expression for~$q$ is the reason why many statement that hold in vacuum strongly ressemble those found in stiff fluid.

\subsection*{Acknowledgements}

The author wishes to express gratitude 
towards Hans Ringström for fruitful discussions and ongoing supervision 
as well as numerous comments on the article.

The author would like to acknowledge the support of the Göran Gustafsson Foundation for
Research in Natural Sciences and Medicine.
This research was supported by the Swedish Research Council,
Reference number 621-2012-3103.

\section{Recollection of basic properties}

The evolution equations~\eqref{eqns_evolutionbianchib}--\eqref{eqn_evolutionomega} are a special case of the setting discussed in~\cite{hewittwainwright_dynamicalsystemsapproachbianchiorthogonalB} and~\cite{radermacher_sccogonbianchibperfectfluidsvacuum}, namely when choosing~$\gamma=2$ and~$\Omega>0$, see Appendix~\ref{appendix_nonstiffevolution}. In this section, we recall a number of results obtained in the general setting which we make use of in the present paper.

\begin{rema}
\label{rema_statespacecompact}
	According to~\cite[Rem.~3.3]{radermacher_sccogonbianchibperfectfluidsvacuum}, the state space of the evolution is compact, and the individual variables are contained in the following intervals:
	\begin{equation}
		\Sigma_+ \in[{-}1,1],\quad
		\tilde\Sigma \in[0,1],\quad
		\tilde A \in[0,1],\quad
		\tilde N \in[0,1],\quad
		\Omega \in[0,1],\quad
		q \in[{-}1,2].
	\end{equation}
	For the stiff fluid evolution equations~\eqref{eqns_evolutionbianchib}--\eqref{eqn_evolutionomega}, the bounds on~$q$ can be further improved: 
	Combining the definition of~$q$, equation~\eqref{eqn_qwithaandn}, with the definition of~$\Omega$, equation~\eqref{eqn_omegageneral}, we find that
	\begin{equation}
	\label{eqn_qwithsigma}
		q=2(\Sigma_+^2+\tilde\Sigma+\Omega).
	\end{equation}
	This yields
	\begin{equation}
		q\in[0,2],
	\end{equation}
	and one also finds
	\begin{equation}
		\tilde A+\tilde N\le1.
	\end{equation}
\end{rema}
\begin{rema}
\label{rema_symmetryevolution}
	The evolution equations~\eqref{eqns_evolutionbianchib}--\eqref{eqn_evolutionomega} as well as the additional expression for~$q$, equation~\eqref{eqn_qwithsigma}, are invariant under the symmetry
	\begin{equation}
		(\Delta,N_+)\mapsto {-}(\Delta,N_+) .
	\end{equation}
\end{rema}
\begin{rema}
	The function
	\begin{equation}
	\label{eqn_fctZ}
		Z=(1+\Sigma_+)^2-\tilde A,
	\end{equation}
	introduced in~\cite{hewittwainwright_dynamicalsystemsapproachbianchiorthogonalB}, satisfies~$Z\ge0$. The set~$Z=0$ is contained in vacuum~$\Omega=0$.
	Restricting to the setting of stiff fluids, \ie $\gamma=2$ and~$\Omega>0$, this function has derivative
	\begin{equation}
	\label{eqn_fctZderivative}
		Z'={-}2(2-q)Z.
	\end{equation}
	For a more detailed discussion, see~\cite[p.~19]{radermacher_sccogonbianchibperfectfluidsvacuum}.
\end{rema}
\begin{rema}
\label{rema_singularconstraintequ}
	At several instances in this paper, we restrict our attention to the evolution \equ s~\eqref{eqns_evolutionbianchib}--\eqref{eqn_definitiontilden} in~$\RR^5$, \ie without assuming the constraint \equ s~\eqref{eqn_constraintgeneralone}--\eqref{eqn_constraintgeneraltwo}. We call this dynamical system the \textit{evolution \equ s in the extended state space}, see also~\cite[Rem~3.5]{radermacher_sccogonbianchibperfectfluidsvacuum}.
	
	When constraining the evolution to the physical state space, \ie to the set defined by the constraint \equ s~\eqref{eqn_constraintgeneralone}--\eqref{eqn_constraintgeneraltwo}, we recover the original evolution equations.
	The first constraint equation is invariant under the evolution~\eqref{eqns_evolutionbianchib}--\eqref{eqn_definitiontilden}. As long as the gradient of its \lhs
	\begin{equation}
	\label{eqn_gradientconstraintequ}
		({-}2\Sigma_+\tilde A\,,\,\frac13(N_+^2-{\bparamk}\tilde A)\,,\,{-}2\Delta\,,\,{-}(\Sigma_+^2+\frac13{\bparamk}\tilde\Sigma)\,,\,\frac23\tilde\Sigma N_+)
	\end{equation}
	does not vanish, the hypersurface defined by the constraint equation~\eqref{eqn_constraintgeneralone} is non-singular and thus a sub\mf. The only exceptions to a non-vanishing gradient are:
	\begin{itemize}
		\item ${\bparamk}>0$, $\Sigma_+=\tilde\Sigma=\Delta=0$, $N_+^2={\bparamk}\tilde A$, which defines an invariant set of dimension one.
		\item ${\bparamk}=0$, $\Sigma_+=N_+=\Delta=0$, which defines an invariant set of dimension two.
		\item ${\bparamk}<0$, $\Delta=\tilde A=N_+=0$, ${\bparamk}\tilde\Sigma+3\Sigma_+^2=0$, which defines an invariant set of dimension one.
	\end{itemize}
\end{rema}

\section{Recollection of decay lemmata}

When discussing detailed properties of solutions close to their limit points, we apply two general decay lemmata which we already made use of in~\cite{radermacher_sccogonbianchibperfectfluidsvacuum}, where they appeared as Lemma~4.5 and Lemma~4.7. For completeness, we also state them here.
\begin{lemm}
\label{lemm_decaylemmageneral}
	Consider a positive function $M:\RR\rightarrow(0,\infty)$ satisfying $M'=\functiondecaylemma M$, with $\functiondecaylemma:\RR\rightarrow\RR$ and $\functiondecaylemma(\tau) \rightarrow\iota$ as $\tau\rightarrow{-}\infty$. Then for all~$\eps>0$ there is a~$\tau_\eps>{-}\infty$ such that~$\tau\le \tau_\eps$ implies
	\begin{equation}
		e^{(\iota+\eps)\tau} \le  M(\tau)\le e^{(\iota-\eps)\tau}.
	\end{equation}
\end{lemm}
\begin{lemm}
\label{lemm_decaylemmageneralimproved}
	Consider a positive function $M:\RR\rightarrow(0,\infty)$ satisfying $M'=\functiondecaylemma M$, with $\functiondecaylemma:\RR\rightarrow\RR$ and $\functiondecaylemma(\tau) =\iota+\O(e^{\xi\tau})$ as~$\tau\rightarrow{-}\infty$, for some constants~$\xi>0$, $\iota\in\RR$. Then there are constants~$c_M,C_M>0$ and a~$\tau_0$ such that~$\tau\le \tau_0$ implies
	\begin{equation}
		c_M e^{\iota\tau} \le  M(\tau)\le C_M e^{\iota\tau}.
	\end{equation}
\end{lemm}

\section{The Jacobs set \texorpdfstring{$\stiffalphalimit$}{J}}

In this section, we show that the Jacobs set~$\stiffalphalimit$ given in Definition~\ref{defi_setD} dominates the behaviour of solutions to the evolution equations~\eqref{eqns_evolutionbianchib}--\eqref{eqn_evolutionomega} at early times. It has been shown in~\cite[Prop.~5.3]{hewittwainwright_dynamicalsystemsapproachbianchiorthogonalB} that~$\stiffalphalimit$ contains the~$\alpha$-limit set of such solutions. We give a refined proof of this statement, using a method of proof similar to that of~\cite[Prop.~4.2]{radermacher_sccogonbianchibperfectfluidsvacuum}.
\begin{prop}[Alpha-limit set in stiff fluid]
\label{prop_alphalimitsets_stiff}
	If $\Gamma(\tau)=(\Sigma_+,\tilde\Sigma,\Delta,\tilde A,N_+)(\tau)$, $\tau\in\RR$, is a solution to \equs ~\eqref{eqns_evolutionbianchib}--\eqref{eqn_evolutionomega}, then the $\alpha$-limit set of~$\Gamma$ satisfies
	\begin{equation}
		\alpha(\Gamma)\subset\stiffalphalimit,
	\end{equation}
	\ie for every sequence~$\tau_n\rightarrow{-}\infty$ the accumulation points of~$\Gamma(\tau_n)$ are contained in~$\stiffalphalimit$.
\end{prop}
\begin{proof}
	For a given solution to the evolution equations, we consider the function~$Z$ defined as in~\eqref{eqn_fctZ} and conclude from its derivative~\eqref{eqn_fctZderivative} that~$Z$ is strictly monotone decreasing if~$q<2$ and~$Z>0$. As a consequence, the~$\alpha$-limit set of the solution is contained in the union of~$\{q=2\}$ and~$\{Z=0\}$. The first set characterises~$\stiffalphalimit$.
	
	Suppose there is an~$\alpha$-limit point satisfying~$Z=0$. Then one can find a sequence of times~$\tau_n\rightarrow{-}\infty$ such that
	\begin{equation}
		Z(\tau_n)\le\frac1n.
	\end{equation}
	As~$Z\ge0$, monotonicity of~$Z$ implies that~$Z$ vanishes along the whole solution. However, the set~$\{Z=0\}$ is contained in vacuum~$\Omega=0$, a contradiction.
\end{proof}

In the next proposition we strengthen this result by showing that every solution has a unique~$\alpha$-limit point, \ie converges.
We start by showing that the density parameter~$\Omega$ converges.
\begin{lemm}
	Consider a solution to equations~\eqref{eqns_evolutionbianchib}--\eqref{eqn_evolutionomega}.
	Then~$\Omega$ is monotone decreasing and converges to a positive value, say~$\Omega\rightarrow\omega_0>0$ as~$\tau\rightarrow{-}\infty$.
\end{lemm}
\begin{proof}
	As~$q\le2$, the evolution equation of~$\Omega$, equation~\eqref{eqn_evolutionomega}, gives~$\Omega'\le0$. The statement on the limit value follows due to boundedness of~$\Omega$.
\end{proof}
\begin{prop}
\label{prop_convergencetoD}
	Consider a solution to \equs ~\eqref{eqns_evolutionbianchib}--\eqref{eqn_evolutionomega},
	then there is~$\slimit\in({-}1,1)$,~$\stildelimit\in[0,1)$ with~$\slimit^2+\stildelimit<1$ such that
	\begin{equation}
		\lim_{\tau\rightarrow{-}\infty}(\Sigma_+,\tilde\Sigma,\Delta, N_+,\tilde A)=(\slimit,\stildelimit,0,0,0).
	\end{equation}
\end{prop}
\begin{proof}
	From the previous lemma, we conclude that we can choose~$\tau_0$ suitably negative such that~$\Omega$ is bounded from below by a positive number on~$({-}\infty,\tau_0)$. Combining this with the evolution equation~\eqref{eqn_evolutionomega} for~$\Omega$ and the knowledge that~$\Omega$ converges as~$\tau\rightarrow{-}\infty$, we conclude that $q-2$ is integrable on~$({-}\infty,\tau_0)$.
	As~$q'$ is bounded due to the boundedness of the state space, Remark~\ref{rema_statespacecompact}, this implies~$q\rightarrow2$.
	Similarly, the derivatives~$\tilde A'$ and~$\tilde N'$ are bounded, and equation~\eqref{eqn_qwithaandn} yields integrability of~$\tilde A$ and~$\tilde N$ on~$({-}\infty,\tau_0)$, therefore both variables converge to~$0$.

	In combination with the integrabilities found so far and boundedness of the state space, integrability of~$N_+^2$ follows from the definition of~$\tilde N$, equation~\eqref{eqn_definitiontilden}, and integrability of~$\Delta^2$ from the constraint equation~\eqref{eqn_constraintgeneralone}. The evolution equations of~$N_+^2$ and~$\Delta^2$ are bounded, which yields convergence of the two squares, and thus of~$N_+$ and~$\Delta$, to zero.

	For the remaining quantities~$\Sigma_+$ and~$\tilde\Sigma$, one finds that all terms in the evolution equations are integrable, where integrability of~$\Delta N_+$ is a consequence of the Cauchy--Schwarz inequality. Consequently, these variables converge as well, to values~$\slimit$ and~$\stildelimit$. The bound on~$\slimit^2+\stildelimit$ follows from the properties of the~$\alpha$-limit set, Prop.~\ref{prop_alphalimitsets_stiff}.
\end{proof}

\section{Convergence properties and asymptotic decay towards the Jacobs set \texorpdfstring{$\stiffalphalimit$}{J}}

\label{section_asymptoticconvergence}

In this section, we discuss in detail the asymptotic properties of solutions converging to~$\stiffalphalimit$. 
To this end, we consider a solution to~\eqref{eqns_evolutionbianchib}--\eqref{eqn_evolutionomega} and assume that for~$\tau\rightarrow{-}\infty$ the solution converges to a point~$(\slimit,\stildelimit,0,0,0)\in\stiffalphalimit$, \ie $\slimit\in({-}1,1)$,~$\stildelimit\in[0,1)$ and~$\slimit^2+\stildelimit<1$. In this setting, we determine the asymptotic convergence properties of the individual variables.
Note that assuming convergence to a limit point in the Jacobs set~$\stiffalphalimit$ is no loss of generality due to Proposition~\ref{prop_convergencetoD}.

In terms of technique, this section is the heart of the present article. In the following sections, we discuss convergence to the different parts of~$\stiffalphalimit$ in a more global sense, and from a dynamical systems point of view. All our arguments are, however, rooted in statements achieved in the present section.

The steps of our discussion in large parts mirror those of~\cite[Section~6]{radermacher_sccogonbianchibperfectfluidsvacuum}, where we assumed convergence to a point on the arc~$(\slimit,1-\slimit^2,0,0,0)$ for~$\slimit\in[{-}1,1]$. This arc is called the Kasner parabola~$\kasnerparabola$. Although this parabola is not contained in the Jacobs set~$\stiffalphalimit$ but in its boundary, the methods of proof can be transferred. Compared with~\cite{radermacher_sccogonbianchibperfectfluidsvacuum}, the limit point in~$\stiffalphalimit$ depends on two parameters, $\slimit$ and~$\stildelimit$, and this is reflected in many of the statements, as we find exponential convergence rates which depend on both values. Additional difficulties arise for solutions converging to a point on the~$\Sigma_+$-axis. 

As the proofs of some of the statements in the present section are very similar to those given in our previous paper, we refrain from stating them explicitly. Instead, we expect the reader to have~\cite{radermacher_sccogonbianchibperfectfluidsvacuum} at hand.

The main statement we achieve in this section is the following proposition, which gives detailed convergence rates of the individual variables for solutions converging to a limit point situated in either~$\stiffalphalimitminus$ or~$\stiffalphalimitzero$. In particular, we find that if such a solution satisfies~$\tilde A>0$, \ie if the solution is of Bianchi class~B, then a specific relation between the limit point~$(\slimit,\stildelimit,0,0,0)$ and the parameter~$\bparamk$ has to hold.
The statement should be compared with~\cite[Prop.~6.2]{radermacher_sccogonbianchibperfectfluidsvacuum}. For the proof, a number of preliminary work is needed, which we carry out in the following.
\begin{prop}
\label{prop_convergencemain_unstablepart}
	\considerconvtoD.
	If $(1+\slimit)^2\le3\stildelimit$, then
	\begin{equation}
		\tilde A(3\slimit^2+{\bparamk}\stildelimit)=0
	\end{equation}
	along the whole solution, and
	\begin{align}
		\Sigma_+={}&\slimit+\O(e^{(4+4\slimit-\eps)\tau}),\\
		\tilde\Sigma={}&\stildelimit+\O(e^{(4+4\slimit-\eps)\tau}),\\
		\tilde N={}&\O(e^{(4+4\slimit-\eps)\tau}),\\
		q={}&2+\O(e^{(4+4\slimit-\eps)\tau}),
	\end{align}
	as $\tau\rightarrow{-}\infty$, for every $\eps>0$.
	Furthermore, the following properties hold:
	\begin{itemize}
		\item Either $\Delta=0=N_+$ holds along the whole solution, or $\Delta N_+>0$ for sufficiently negative times. In the latter case, for all $\eps>0$ there is a $\tau_\eps>{-}\infty$ such that $\tau \le\tau_\eps$ implies
		\begin{align}
			e^{(2+2\slimit+2\sqrt{3\stildelimit}+\eps)\tau}
			\le{}& \absval{\Delta} \le e^{(2+2\slimit+2\sqrt{3\stildelimit}-\eps)\tau},\\
			e^{(2+2\slimit+2\sqrt{3\stildelimit}+\eps)\tau}
			\le{}&\absval{N_+} \le e^{(2+2\slimit+2\sqrt{3\stildelimit}-\eps)\tau}.
		\end{align}
		\item Either $\tilde A=0$ holds along the whole solution, or for all $\eps>0$ there is a $\tau_\eps>{-}\infty$ such that $\tau\le \tau_\eps$ implies
		\begin{equation}
			e^{(4+4\slimit+\eps)\tau} \le\tilde A\le e^{(4+4\slimit-\eps)\tau}.
		\end{equation}
	\end{itemize}
\end{prop}

The following two statements replace Lemma~6.5 and Lemma~6.6 of~\cite{radermacher_sccogonbianchibperfectfluidsvacuum} in the case of a stiff fluid. Note that the decay of~$\tilde A$ only depends on~$\slimit$, not on~$\stildelimit$.
\begin{lemm}
\label{lemm_decaytildea}
	\considerconvtoD. Then either $\tilde A=0$ along the whole solution, or there exists for every $\eps>0$ a $\tau_\eps>{-}\infty$ such that $\tau\le\tau_\eps$ implies
	\begin{equation}
		e^{(4+4\slimit+\eps)\tau}\le \tilde A(\tau)\le e^{(4+4\slimit-\eps)\tau}.
	\end{equation}
\end{lemm}
\begin{proof}
	This follows from the evolution equation~\eqref{eqns_evolutionbianchib} for~$\tilde A$ as a  direct consequence of Lemma~\ref{lemm_decaylemmageneral}, as~$\tau\rightarrow{-}\infty$ implies~$q+2\Sigma_+\rightarrow2+2\slimit$ by assumption.
\end{proof}
\begin{lemm}
\label{lemm_functionintegraltwominusq}
	\considerconvtoD. Then the function
	\begin{equation}
		\integraltwominusq\coloneqq\int_{{-}\infty}^\tau(2-q)ds
	\end{equation}
	is well-defined.
\end{lemm}
\begin{proof}
	We have argued in the proof of Proposition~\ref{prop_convergencetoD} that~$q-2$ is integrable on~$({-}\infty,\tau_0)$, for sufficiently negative~$\tau_0$. Thus~$\integraltwominusq$ is well-defined.
\end{proof}
\begin{lemm}
\label{lemm_deltanplussigns_stiff}
	\considerconvtoD, assume additionally that~$\stildelimit>0$. Then one of the following statements holds:
	\begin{enumerate}
		\item $\Delta=N_+=0$ along the whole solution,
		\item There is $\tau_1\in\RR$ \st $\Delta N_+(\tau)>0$ for all $\tau\le\tau_1$,
		\item There is $\tau_1\in\RR$ \st $\Delta N_+(\tau)<0$ for all $\tau\le\tau_1$.
	\end{enumerate}
\end{lemm}
\begin{proof}
	The proof of this statement uses the same steps as that of~\cite[Lemma~6.4]{radermacher_sccogonbianchibperfectfluidsvacuum}, which is centered around the evolution equation
	\begin{equation}
		(\Delta N_+)'=6\Delta^2+(3q+4\Sigma_+-2)\Delta N_++2(\tilde\Sigma-\tilde N)N_+^2.
	\end{equation}
	By assuming positivity of~$\stildelimit$, we ensure that the factors in front of~$\Delta^2$ and~$N_+^2$ are positive for sufficiently negative times. This suffices to conclude.
\end{proof}

As a first step towards exponential decay rates for~$\Delta$ and~$N_+$ individually, we discuss a linear combination of these two variables to which we can apply Lemma~\ref{lemm_decaylemmageneral}. The approach is similar to that of~\cite[Lemma~6.7]{radermacher_sccogonbianchibperfectfluidsvacuum}, and as was the case there, we find that the exponent in the decay rate coincides with one of the eigenvalues of the linearised evolution equations in the Jacobs set~$\stiffalphalimit$.

We have to restrict to the case~$\stildelimit>0$, as we apply the previous lemma to make statements about the signs of~$\Delta$ and~$N_+$. The case~$\stildelimit=0$ is treated further down in Lemma~\ref{lemm_decaystiff_axis}.
\begin{lemm}
\label{lemm_decaydeltanplus_stiff_notaxis}
	\considerconvoffaxisD. If $\Delta N_+(\tau)>0$ for all $\tau\le\tau_0$, then for every $\eps>0$ there exists $\tau_\eps>{-}\infty$ such that $\tau\le\tau_\eps$ implies
	\begin{equation}
		e^{(2+2\slimit+2\sqrt{3\stildelimit}+\eps)\tau}\le \absval{\Delta+\sqrt{\frac{\stildelimit}3}N_+}\le e^{(2+2\slimit+2\sqrt{3\stildelimit}-\eps)\tau}.
	\end{equation}
	If on the other hand $\Delta N_+(\tau)<0$ for all $\tau\le\tau_0$, then for every $\eps>0$ there exists $\tau_\eps>{-}\infty$ such that $\tau\le\tau_\eps$ implies
	\begin{equation}
		e^{(2+2\slimit-2\sqrt{3\stildelimit}+\eps)\tau}\le \absval{\Delta-\sqrt{\frac{\stildelimit}3}N_+}\le e^{(2+2\slimit-2\sqrt{3\stildelimit}-\eps)\tau}.
	\end{equation}
\end{lemm}
\begin{proof}
	For notational convenience, we set
	\begin{equation}
	\label{eqn_definitionrfactor_stiff}
		\rfactorstiff\coloneqq\sqrt{\frac{\stildelimit}3}
	\end{equation}
	and note that
	\begin{equation}
	\label{eqn_propertyrfactor}
		2\stildelimit=6\rfactorstiff^2.
	\end{equation}
	
	We start with the case~$\Delta N_+>0$. The evolution equations~\eqref{eqns_evolutionbianchib} yield 
	\begin{equation}
		(\Delta+\rfactorstiff N_+)'=(2q+2\Sigma_+-2+6\rfactorstiff)\Delta
		+(\rfactorstiff q+2\rfactorstiff\Sigma_++2\tilde\Sigma-2\tilde N)N_+,
	\end{equation}
	and the assumption on the convergence of the solution gives the limits of the brackets,
	\begin{align}
		&\lim_{\tau\rightarrow{-}\infty}(2q+2\Sigma_+-2+6\rfactorstiff) =2+2\slimit+6\rfactorstiff,\\
		&\lim_{\tau\rightarrow{-}\infty}(\rfactorstiff q+2\rfactorstiff\Sigma_++2\tilde\Sigma-2\tilde N) =2\rfactorstiff+2\rfactorstiff\slimit+2\stildelimit.
	\end{align}
	For this reason, there are functions~$f_1,f_2$ converging to~$0$ as~$\tau\rightarrow{-}\infty$ and such that
	\begin{equation}
	\label{eqn_evolutionrdeltaplusnplus}
	\begin{split}
		(\Delta+\rfactorstiff N_+)'
			={}&(2+2\slimit+6\rfactorstiff+f_1)\Delta+(2\rfactorstiff+2\rfactorstiff\slimit+6\rfactorstiff^2+f_2)N_+\\
			={}&(2+2\slimit+6\rfactorstiff+\frac{f_1\Delta+f_2N_+}{\Delta+\rfactorstiff N_+})(\Delta+\rfactorstiff N_+)\\
			={}&(2+2\slimit+2\sqrt{3\stildelimit}+\frac{f_1\Delta+f_2N_+}{\Delta+\rfactorstiff N_+})(\Delta+\rfactorstiff N_+).
	\end{split}
	\end{equation}
	In this computation, we have made use of the fact that~$\Delta$ and~$N_+$ have the same sign. The quotient in the bracket vanishes asymptotically, and Lemma~\ref{lemm_decaylemmageneral} yields the decay of~$\Delta+\rfactorstiff N_+$ in case both~$\Delta$ and~$N_+$ are positive. If both are negative, the statement follows due to the invariance of the evolution equations~\eqref{eqns_evolutionbianchib}--\eqref{eqn_evolutionomega} under a change of sign in these two variables, see Remark~\ref{rema_symmetryevolution}.
	
	The case~$\Delta N_+<0$ is treated by replacing every occurence of~$\rfactorstiff$ by~${-}\rfactorstiff$.
\end{proof}

In the following, we determine convergence rates for all variables, depending on the sign which~$\Delta N_+$ has asymptotically. We start with the case of $\Delta N_+>0$ and find that the convergence rates for~$\Sigma_+$, $\tilde\Sigma$, $\tilde A$, $\tilde N$ and~$q$ are the same as in the vacuum case of~\cite[Lemma~6.8]{radermacher_sccogonbianchibperfectfluidsvacuum}. In particular, these rates only depend on~$\slimit$, while the decay rates for~$\Delta$ and~$N_+$ also depend on~$\stildelimit$.
As was the case in non-stiff fluid, the arguments revolve around the constraint \equ ~\eqref{eqn_constraintgeneralone} written in the form
\begin{equation}
\label{eqn_constraintforproofs}
	\tilde\Sigma N_+^2-3\Delta^2=(3\Sigma_+^2+{\bparamk}\tilde\Sigma)\tilde A
\end{equation}
\begin{lemm}
\label{lemm_deltanplussamesign_stiff}
	\considerconvoffaxisD.
	Assume that $\Delta N_+(\tau)>0$ for all $\tau\le\tau_0$. Then
	\begin{align}
		\Sigma_+={}&\slimit+\O(e^{(4+4\slimit-\eps)\tau}),\\
		\tilde\Sigma={}&\stildelimit+\O(e^{(4+4\slimit-\eps)\tau}),\\
		\Delta={}&\O(e^{(2+2\slimit+2\sqrt{3\stildelimit}-\eps)\tau}),\\
		N_+={}&\O(e^{(2+2\slimit+2\sqrt{3\stildelimit}-\eps)\tau}),\\
		\tilde N={}&\O(e^{(4+4\slimit-\eps)\tau}),\\
		q={}&2+\O(e^{(4+4\slimit-\eps)\tau}),
	\end{align}
	as $\tau\rightarrow{-}\infty$, for every $\eps>0$.
	Furthermore,
	\begin{equation}
		\tilde A(3\slimit^2+{\bparamk}\stildelimit)=0
	\end{equation}
	holds along the whole solution, and $\tilde A>0$ implies that
	\begin{equation}
		3\Sigma_+^2+{\bparamk}\tilde\Sigma =\O(e^{(4+4\slimit+4\sqrt{3\stildelimit}-\eps)\tau})
	\end{equation}
	as $\tau\rightarrow{-}\infty$, for every $\eps>0$.

	Further, there is for each $\eps>0$ a $\tau_\eps>{-}\infty$ such that $\tau\le \tau_\eps$ implies
	\begin{align}
	\label{eqn_lowerbounddelta}
		e^{(2+2\slimit+2\sqrt{3\stildelimit}+\eps)\tau}\le {}&\absval{\Delta},\\
	\label{eqn_lowerboundnplus}
		e^{(2+2\slimit+2\sqrt{3\stildelimit}+\eps)\tau}\le {}&\absval{N_+}.
	\end{align}
\end{lemm}
\begin{rema}
	Note that we excluded the case~$\stildelimit=0$ in the statement. The reason is that we wish to apply Lemma~\ref{lemm_decaydeltanplus_stiff_notaxis}, which only holds for~$\stildelimit>0$. The case~$\stildelimit=0$ is treated in Lemma~\ref{lemm_decaystiff_axis}.
\end{rema}

\begin{rema}
\label{rema_auxiliarydecayconstraintexpr_stiff}
	It follows from Lemma~\ref{lemm_deltanplussamesign_stiff} that in this setting
	\begin{equation}
		\frac13\tilde\Sigma N_+^2-\Delta^2=(\Sigma_+^2+\frac {\bparamk}3\tilde\Sigma)\tilde A=\O(e^{(8+8\slimit+4\sqrt{3\stildelimit}-\eps)\tau})
	\end{equation}
	holds for both $\tilde A=0$ and $\tilde A>0$.
\end{rema}
\begin{proof}[Proof of Lemma~\ref{lemm_deltanplussamesign_stiff}]
	The upper bounds for~$\Delta$ and~$N_+$ follow from Lemma~\ref{lemm_decaydeltanplus_stiff_notaxis}, where we determined that for sufficiently negative~$\tau$
	\begin{equation}
		\absval{\Delta+\rfactorstiff N_+}\le e^{(2+2\slimit+2\sqrt{3\stildelimit}-\eps)\tau},
	\end{equation}
	with~$\rfactorstiff$ as in \equ ~\eqref{eqn_definitionrfactor_stiff} and therefore positive.
	In combination with the decay $\tilde A=\O(e^{(4+4\slimit-\eps)\tau})$ determined in Lemma~\ref{lemm_decaytildea}, this implies
	\begin{align}
		\tilde N={}&\frac13(N_+^2-{\bparamk}\tilde A)=\O(e^{(4+4\slimit-\eps)\tau}),\\
		q={}&2(1-\tilde A-\tilde N)=2+\O(e^{(4+4\slimit-\eps)\tau}),
	\end{align}
	where we used equation~\eqref{eqn_qwithaandn} for~$q$.
	Inserting these decay rates in the evolution \equs~\eqref{eqns_evolutionbianchib} implies that
	\begin{equation}
		\Sigma_+'=\O(e^{(4+4\slimit-\eps)\tau}),\qquad \tilde\Sigma'=\O(e^{(4+4\slimit-\eps)\tau}),
	\end{equation}
	and together with the assumption on the limit point gives the asymptotic behaviour for~$\Sigma_+$ and~$\tilde\Sigma$.

	\smallskip

	Applying the convergence rates obtained so far, the constraint \equ\ written as in~\eqref{eqn_constraintforproofs} reads
	\begin{equation}
	\label{eqn_constraintinproofsamesign}
		\tilde\Sigma N_+^2-3\Delta^2=(3\Sigma_+^2+{\bparamk}\tilde\Sigma)\tilde A=(3\slimit^2+{\bparamk}\stildelimit+\O(e^{(4+4\slimit-\eps)\tau}))\tilde A,
	\end{equation}
	where the \lhs\ is of order $\O(e^{(4+4\slimit+4\sqrt{3\stildelimit}-\eps)\tau})$. The \rhs\ can either vanish altogether, or it consists of two non-vanishing factors. In the latter case, the decay of~$\tilde A$ is at most as fast as $e^{(4+4\slimit+\eps)\tau}$, see Lemma~\ref{lemm_decaytildea}. For consistency reasons,  equation~\eqref{eqn_constraintinproofsamesign} can therefore only hold if~$\tilde A=0$ or $3\slimit^2+{\bparamk}\stildelimit=0$.
	
	We discuss the first factor on the \rhs\ of~\eqref{eqn_constraintinproofsamesign} in more detail: The evolution equations~\eqref{eqns_evolutionbianchib} together with the definition of~$\tilde N$, equation~\eqref{eqn_definitiontilden}, give
	\begin{equation}
		(3\Sigma_+^2+{\bparamk}\tilde\Sigma)'=6\Sigma_+\Sigma_+'+{\bparamk}\tilde\Sigma'=2(q-2)(3\Sigma_+^2+{\bparamk}\tilde\Sigma)-4N_+(\Sigma_+N_++{\bparamk}\Delta).
	\end{equation}
	Using the function~$\integraltwominusq(\tau)=\int_{{-}\infty}^\tau(2-q)ds$ from Lemma~\ref{lemm_functionintegraltwominusq}, we can compute from this
	\begin{equation}
		\frac{d}{d\tau}(e^{2\integraltwominusq}(3\Sigma_+^2+{\bparamk}\tilde\Sigma))
		={-}e^{2\integraltwominusq}4N_+(\Sigma_+N_++{\bparamk}\Delta),
	\end{equation}
	and therefore
	\begin{align}
		e^{2\integraltwominusq}(3\Sigma_+^2+{\bparamk}\tilde\Sigma)
			={}& 3\slimit^2+{\bparamk}\stildelimit- \int_{{-}\infty}^\tau e^{2\integraltwominusq}4N_+(\Sigma_+N_++{\bparamk}\Delta)ds\\
			={}& 3\slimit^2+{\bparamk}\stildelimit +\O(e^{(4+4\slimit+4\sqrt{3\stildelimit}-\eps)\tau}),
	\end{align}
	where we have applied the decay rates obtained above. Assuming that $\tilde A>0$, the constant term has been shown to vanish, and we conclude that
	\begin{equation}
		3\Sigma_+^2+{\bparamk}\tilde\Sigma=\O(e^{(4+4\slimit+4\sqrt{3\stildelimit}-\eps)\tau}).
	\end{equation}
	As a consequence, the term~$(\Sigma_+^2+\frac {\bparamk}3\tilde\Sigma)\tilde A$ decays exponentially to order at least $8+8\slimit+4\sqrt{3\stildelimit}-\eps$, both in case~$\tilde A=0$ and in case~$\tilde A>0$: In the first case, the term vanishes identically, while in the second one we have determined the decay properties of both factors individually.

	\smallskip

	To conclude the proof, one needs to show the lower bound for~$\Delta$ and~$N_+$. 
	The proof mirrors that of~\cite[Lemma~6.8]{radermacher_sccogonbianchibperfectfluidsvacuum}: From the last argument together with Lemma~\ref{lemm_decaydeltanplus_stiff_notaxis} one finds that 
	\begin{equation}
		\frac13\tilde\Sigma(\frac{N_+}{\Delta+\rfactorstiff N_+})^2-(\frac{\Delta}{\Delta+\rfactorstiff N_+})^2=\frac{(\Sigma_+^2+\frac {\bparamk}3\tilde\Sigma)\tilde A}{({\Delta+\rfactorstiff N_+})^2}\rightarrow0
	\end{equation}
	as $\tau\rightarrow{-}\infty$, and due to our assumption~$\stildelimit>0$, we know that~$\tilde\Sigma$ is bounded away from zero for sufficiently negative times.
\end{proof}

\begin{lemm}
\label{lemm_deltanpluszero_stiff}
	\considerconvtoD.
	Assume that
	$\Delta=N_+=0$ along the solution.
	Then
	\begin{equation}
		\tilde A(3\slimit^2+{\bparamk}\stildelimit)=0
	\end{equation}
	and
	\begin{align}
		\Sigma_+={}&\slimit+\O(e^{(4+4\slimit-\eps)\tau}),\\
		\tilde\Sigma={}&\stildelimit+\O(e^{(4+4\slimit-\eps)\tau}),\\
		\tilde N={}&\O(e^{(4+4\slimit-\eps)\tau}),\\
		q={}&2+\O(e^{(4+4\slimit-\eps)\tau}),\\
	\end{align}
	as $\tau\rightarrow{-}\infty$, for every $\eps>0$.
\end{lemm}
\begin{proof}
	The proof goes along the lines of that of~\cite[Lemma~6.10]{radermacher_sccogonbianchibperfectfluidsvacuum}: As~$\Delta=N_+=0$, the decay of~$\tilde A$ from Lemma~\ref{lemm_decaytildea} determines the asymptotic behaviour of~$\tilde N$ and~$q$ as well as~$\Sigma_+'$ and~$\tilde\Sigma'$.
\end{proof}

For limit points contained in~$\stiffalphalimitminus$ or~$\stiffalphalimitzero$, \ie satisfying $(1+\slimit)^2\le3\tilde\slimit$, the case $\Delta N_+<0$ can be excluded. This statement should be compared with~\cite[Lemma~6.11]{radermacher_sccogonbianchibperfectfluidsvacuum}, where a negative sign of~$\Delta N_+$ could be excluded upon convergence towards Kasner points to the left of Taub~2.
\begin{lemm}
\label{lemm_deltanplusoppositesign_stiff}
	\considerconvtoD.
	If $\Delta N_+(\tau_k)<0$ for a sequence $\tau_k\rightarrow{-}\infty$, then $(1+\slimit)^2>3\tilde\slimit$.
\end{lemm}

\begin{proof}
	If~$\stildelimit=0$, then~$(1+\slimit)^2>3\stildelimit$ holds trivially, as~$\slimit\in({-}1,1)$. We can therefore assume that~$\stildelimit>0$.
	Due to Lemma~\ref{lemm_deltanplussigns_stiff}, we then know that~$\Delta N_+$ has a fixed sign for sufficiently negative~$\tau$, hence it is enough to show that solutions with $\Delta N_+(\tau)<0$ for $\tau\le\tau_0$ cannot converge to a Kasner point with $(1+\slimit)^2\le3\tilde\slimit$.

	Consider first the case $(1+\slimit)^2<3\tilde\slimit$ and recall from Lemma~\ref{lemm_decaydeltanplus_stiff_notaxis} that
	\begin{equation}
		\absval{\Delta-\rfactorstiff N_+}\ge e^{(2+2\slimit-2\sqrt{3\stildelimit}+\eps)\tau}
	\end{equation}
	for~$\tau$ sufficiently negative and~$\rfactorstiff$ as in equation~\eqref{eqn_definitionrfactor_stiff}.
	The bracket in the exponent is strictly negative for sufficiently small~$\eps$, 
	and as~$\Delta$ and~$-\rfactorstiff N_+$ have the same sign, this contradicts the fact that~$\Delta$ and~$N_+$ converge to zero.

	Let us therefore assume that $(1+\slimit)^2=3\tilde\slimit$. This implies the special value $3\rfactorstiff=1+\slimit$ and therefore
	\begin{equation}
		\absval{\Delta-\rfactorstiff N_+}\ge e^{\eps\tau}
	\end{equation}
	for~$\tau$ sufficiently negative. 
	The following argument shows that this estimate holds for~$\Delta$ and~$N_+$ individually: We reformulate the constraint \equ ~\eqref{eqn_constraintforproofs} into
	\begin{align}
		(3\Sigma_+^2+{\bparamk}\tilde\Sigma)\tilde A={}&\tilde\Sigma N_+^2-3\Delta^2\\
			={}&(\tilde\Sigma N_+-3\rfactorstiff^2 N_++3\rfactorstiff^2 N_+-3\rfactorstiff\Delta)N_++3(\rfactorstiff N_+-\Delta)\Delta.
	\end{align}
	and divide the \rhs\ by $3(\rfactorstiff N_+-\Delta)$ to obtain
	\begin{equation}
		(\frac{\tilde\Sigma-3\rfactorstiff^2}{3(\rfactorstiff N_+-\Delta)}N_++\rfactorstiff)N_++\Delta=(\rfactorstiff+f_3)N_++\Delta,
	\end{equation}
	with an asymptotically vanishing function~$f_3$.
	Dividing the \lhs\ by the same expression yields exponential decay to order $\O(e^{(4+4\slimit-2\eps)\tau})$ due to the decay of~$\tilde A$, Lemma~\ref{lemm_decaytildea}. Consequently,
	\begin{equation}
		\absval{N_+}\ge e^{\eps\tau},\qquad \absval{\Delta}\ge e^{\eps\tau}.
	\end{equation}
	Inserting these slow decay rates into the definitions of~$\tilde N$ and~$q$, equations~\eqref{eqn_definitiontilden} and~\eqref{eqn_qwithaandn},
	we even find
	\begin{equation}
		\tilde N\ge e^{\eps\tau},\qquad 2-q\ge e^{\eps\tau}.
	\end{equation}

	\smallskip

	In the rest of the proof we aim at constructing a contradiction to the slow decay of~$N_+$. The arguments are similar to those used in the proof of~\cite[Lemma~6.11]{radermacher_sccogonbianchibperfectfluidsvacuum}, but several of the convergence rates and limiting values differ, as~$\gamma=2$ and the limit point is~$(\slimit,\stildelimit,0,0,0)$ in stiff fluid solutions. We give several intermediate steps but refer to the proof in the non-stiff setting for the explanations.
	
	Taylor expansion of the square root applied to the constraint equation~\eqref{eqn_constraintforproofs} reveals that
	\begin{equation}
	\label{eqn_proofoppositesignsdelta}
		\Delta={-}\frac{\sqrt3}{3}\tilde\Sigma^{\nicefrac12}N_++\O(e^{(4+4\slimit-2\eps)\tau})
	\end{equation}
	due to the assumption on the sign of~$\Delta N_+$,
	and combining this with equations~\eqref{eqn_definitiontilden} and~\eqref{eqn_qwithaandn} yields
	\begin{equation}
		\Delta N_+={-}\frac{\sqrt3}{2}(2-q)\tilde\Sigma^{\nicefrac12} +\O(e^{(4+4\slimit-2\eps)\tau}).
	\end{equation}
	
	Next, we discuss the convergence behaviour of~$\Sigma_+$ and~$\tilde\Sigma$. Making use of the function~$\integraltwominusq$ from Lemma~\ref{lemm_functionintegraltwominusq}, we find that
	\begin{align}
		\Sigma_++1
		={}&(\slimit+1)e^{-\integraltwominusq(\tau)}+\O(e^{(4+4\slimit-2\eps)\tau}),\\
		\tilde\Sigma^{\nicefrac12}-\sqrt3
		={}&(\sqrt{\stildelimit}-\sqrt3) e^{-\integraltwominusq(\tau)}+\O(e^{(4+4\slimit-2\eps)\tau}).
	\end{align}
	We can eliminate~$\integraltwominusq$ using a suitable linear combination of these two expressions and find, recalling the relation between~$\slimit$ and~$\stildelimit$, that
	\begin{align}
	\label{eqn_proofoppositesignscombination}
		\O(e^{(4+4\slimit-2\eps)\tau})={}&
		(\sqrt{\stildelimit}-\sqrt3)(\Sigma_++1)-(\slimit+1)(\tilde\Sigma^{\nicefrac12}-\sqrt3)\\
		={}&(\sqrt{\stildelimit}-\sqrt3)\Sigma_+-(\slimit+1)\tilde\Sigma^{\nicefrac12}+\sqrt{\stildelimit}+\sqrt3\slimit\\
		={}&\frac{\sqrt3}3(\slimit-2)\Sigma_+-(\slimit+1)\tilde\Sigma^{\nicefrac12}+\frac{\sqrt3}3(1+4\slimit).
	\end{align}

	With this, we return to the evolution of~$N_+$. Due to its slow decay found above in combination with \equs ~\eqref{eqn_proofoppositesignsdelta} and~\eqref{eqn_proofoppositesignscombination}, we see
	\begin{align}
		N_+'={}&(q+2\Sigma_+)N_++6\Delta\\
			={}&(q+2\Sigma_+-2\sqrt3\tilde\Sigma^{\nicefrac12})N_++\O(e^{(4+4\slimit-2\eps)\tau})\\
			={}&(q+2\Sigma_+-\frac{2(\slimit-2)}{\slimit+1}\Sigma_+-\frac{2(1+4\slimit)}{\slimit+1}+\O(e^{(4+4\slimit-2\eps)\tau}))N_++\O(e^{(4+4\slimit-2\eps)\tau})\\
			={}&(q-2+\frac{6}{\slimit+1}(\Sigma_+-\slimit)+\O(e^{(4+4\slimit-3\eps)\tau}))N_+,
	\end{align}
	and conclude, using the properties of~$\integraltwominusq$ from Lemma~\ref{lemm_functionintegraltwominusq}, that there is a function~$f_4$ which is integrable on $({-}\infty,0)$ and satisfies
	\begin{equation}
		N_+'=(\frac{6}{\slimit+1}(\Sigma_+-\slimit)+f_4)N_+.
	\end{equation}
	We set
	\begin{equation}
		F_4(\tau)\coloneqq\int_{{-}\infty}^\tau f_4(s)ds
	\end{equation}
	and compute that
	\begin{equation}
		\frac{d}{d\tau}(e^{-F_4}N_+)=\frac{6}{\slimit+1}(\Sigma_+-\slimit)e^{-F_4}N_+.
	\end{equation}
	As
	\begin{equation}
		\Sigma_+'=(q-2)\Sigma_+-2\tilde N=(q-2)(\Sigma_++1)+2\tilde A
	\end{equation}
	due to equation~\eqref{eqn_qwithaandn}, the sign and slow decay of~$q-2$ together with the fast one of~$\tilde A$ and the fact that~$\Sigma_+$ is bounded away from~${-}1$ for sufficiently negative times yields that~$\Sigma_+'\le0$ asymptotically. Hence~$\Sigma_+\le\slimit$, which implies that the function $\absval{e^{-F_4}N_+}$ increases as $\tau\rightarrow{-}\infty$. Therefore, $N_+$ cannot converge to~$0$, a contradiction.
\end{proof}

The statements proven so far put us in a position to show Proposition~\ref{prop_convergencemain_unstablepart}.
\begin{proof}[Proof of Proposition~\ref{prop_convergencemain_unstablepart}]
	As we consider convergence to a point in~$\stiffalphalimitminus$ or~$\stiffalphalimitzero$, we see that~$\stildelimit>0$ has to hold.
	We have shown in Lemma~\ref{lemm_deltanplussigns_stiff} that there are three possible cases regarding the sign of~$\Delta N_+$. The case~$\Delta N_+<0$ asymptotically is excluded by Lemma~\ref{lemm_deltanplusoppositesign_stiff}. The two remaining cases are discussed in Lemma~\ref{lemm_deltanplussamesign_stiff} and Lemma~\ref{lemm_deltanpluszero_stiff}.
\end{proof}

In the previous lemma, we were able to show that solutions converging to~$\stiffalphalimitminus$ or~$\stiffalphalimitzero$ cannot satisfy~$\Delta N_+<0$ asymptotically as~$\tau\rightarrow{-}\infty$. In Proposition~\ref{prop_convergencemain_unstablepart}, we have collected the asymptotic convergence rates of solutions in this setting and seen that the exponential convergence rates coincide with the eigenvalues~$4+4\slimit$ or~$2+2\slimit+2\sqrt{3\stildelimit}$. We now turn our attention to solutions converging towards~$\stiffalphalimitplus$, \ie with a limit point satisfying~$(1+\slimit)^2>3\tilde\slimit$. In this situation, the case~$\Delta N_+<0$ asymptotically can occur. The asymptotic behaviour is governed by the decay rates of~$\Delta$ and~$N_+$, which tend to zero exponentially to order~$2+2\slimit-2\sqrt{3\stildelimit}$.
The following statement should be compared to the vacuum version of~\cite[Lemma~6.12]{radermacher_sccogonbianchibperfectfluidsvacuum}.
\begin{lemm}
\label{lemm_deltanplusrightoftaub2oppositesign_stiff}
	\considerconvoffaxisD.
	Assume that $\Delta N_+(\tau)<0$ for all $\tau\le\tau_0$.
	Then
	\begin{align}
		\Sigma_+={}&\slimit+\O(e^{(4+4\slimit-4\sqrt{3\stildelimit}-\eps)\tau}),\\
		\tilde\Sigma={}&\stildelimit+\O(e^{(4+4\slimit-4\sqrt{3\stildelimit}-\eps)\tau}),\\
		\Delta={}&\O(e^{(2+2\slimit-2\sqrt{3\stildelimit}-\eps)\tau}),\\
		N_+={}&\O(e^{(2+2\slimit-2\sqrt{3\stildelimit}-\eps)\tau}),\\
		\tilde N={}&\O(e^{(4+4\slimit-4\sqrt{3\stildelimit}-\eps)\tau}),\\
		q={}&2+\O(e^{(4+4\slimit-4\sqrt{3\stildelimit}-\eps)\tau}),
	\end{align}
	as $\tau\rightarrow{-}\infty$, for every $\eps>0$.

	Further, there is for each $\eps>0$ a $\tau_\eps>{-}\infty$ such that $\tau\le \tau_\eps$ implies
	\begin{align}
		e^{(2+2\slimit-2\sqrt{3\stildelimit}+\eps)\tau}\le {}&\absval{\Delta},\\
		e^{(2+2\slimit-2\sqrt{3\stildelimit}+\eps)\tau}\le {}&\absval{N_+}.
	\end{align}
\end{lemm}
\begin{rema}
	For the same reason as in Lemma~\ref{lemm_deltanplussamesign_stiff}, we exclude the case~$\stildelimit=0$. This is treated individually in Lemma~\ref{lemm_decaystiff_axis}.
\end{rema}

\begin{proof}
	The upper bounds for~$\Delta$ and~$N_+$ follow from Lemma~\ref{lemm_decaydeltanplus_stiff_notaxis}, where we determined that
	\begin{equation}
		\Delta-\rfactorstiff N_+=\O(e^{(2+2\slimit-2\sqrt{3\stildelimit}-\eps)\tau}),
	\end{equation}
	with $\rfactorstiff$ as in~\eqref{eqn_definitionrfactor_stiff} and therefore positive.
	The decay for~$\tilde A$ has been determined in Lemma~\ref{lemm_decaytildea} independently of the sign of~$\Delta N_+$. As this decay is faster than that of~$\Delta^2$ and~$N_+^2$, the asymptotic behaviour is dominated by the latter two variables.
	More precisely
	\begin{align}
		\tilde N={}&\frac13(N_+^2-{\bparamk}\tilde A)=\O(e^{(4+4\slimit-4\sqrt{3\stildelimit}-\eps)\tau}),\\
		q-2={}&{-}2\tilde A-2\tilde N=\O(e^{(4+4\slimit-4\sqrt{3\stildelimit}-\eps)\tau}),
	\end{align}
	using equations~\eqref{eqn_definitiontilden} and~\eqref{eqn_qwithaandn}.
	With this, the evolution \equ s~\eqref{eqns_evolutionbianchib} for~$\Sigma_+$ and~$\tilde\Sigma$ yields
		\begin{equation}
		\Sigma_+'=\O(e^{(4+4\slimit-4\sqrt{3\stildelimit}-\eps)\tau}),\qquad \tilde\Sigma'=\O(e^{(4+4\slimit-4\sqrt{3\stildelimit}-\eps)\tau}),
	\end{equation}
	which gives the convergence rates for~$\Sigma_+$ and~$\tilde\Sigma$.

	\smallskip

	To prove that the lower bounds for~$\Delta$ and~$N_+$ hold, we proceed similarly as in proof of Lemma~\ref{lemm_deltanplussamesign_stiff}. From the decay properties shown so far, we conclude that the~\rhs\ of
	\begin{equation}
		\frac13\tilde\Sigma(\frac{N_+}{\Delta-\rfactorstiff N_+})^2-(\frac{\Delta}{\Delta-\rfactorstiff N_+})^2=\frac{(\Sigma_+^2+\frac {\bparamk}3\tilde\Sigma)\tilde A}{({\Delta-\rfactorstiff N_+})^2},
	\end{equation}
	converges to~$0$ as $\tau\rightarrow{-}\infty$. Hence, the same is true for the \lhs, and we conclude the lower bounds using the fact that~$\tilde\Sigma$ is bounded away from zero due to~$\stildelimit>0$.
\end{proof}

Convergence towards a limit point in~$\stiffalphalimit$ situated on the~$\Sigma_+$-axis, \ie the case~$\stildelimit=0$, has to be treated individually. Lemma~\ref{lemm_deltanplussigns_stiff} and Lemma~\ref{lemm_decaydeltanplus_stiff_notaxis} do not apply, and we cannot exclude that~$\Delta$ and~$N_+$ change sign infinitely often as~$\tau\rightarrow{-}\infty$. Nonetheless, we can determine asymptotic convergence rates for the individual variables.
\begin{lemm}
\label{lemm_decaystiff_axis}
	\considerconvtoaxisD.
	Then
	\begin{align}
		\Sigma_+={}&\slimit+\O(e^{(4+4\slimit-\eps)\tau}),\\
		\tilde\Sigma={}&\O(e^{(6+6\slimit-2\eps)\tau}),\\
		\Delta={}&\O(e^{(4+4\slimit-\eps)\tau}),\\
		N_+={}&\O(e^{(2+2\slimit-\eps)\tau}),\\
		\tilde N={}&\O(e^{(4+4\slimit-\eps)\tau}),\\
		q={}&2+\O(e^{(4+4\slimit-\eps)\tau}),
	\end{align}
	as $\tau\rightarrow{-}\infty$, for every $\eps>0$. Further, if~$\tilde A>0$, then~$\slimit=0$.
\end{lemm}
\begin{rema}
\label{rema_limitpointsaxis}
	 For Bianchi class~B solutions ($\tilde A>0)$, we find that the only possible limit point on the~$\Sigma_+$-axis is the origin. 
	 
	 For either Bianchi class, we see that~$\tilde\Sigma$ and~$\Delta$ converge faster to their limit value than they do for limit points off the axis, compare Lemma~\ref{lemm_deltanplussamesign_stiff}, Lemma~\ref{lemm_deltanpluszero_stiff} and Lemma~\ref{lemm_deltanplusrightoftaub2oppositesign_stiff}. Note that the decay of these two variables can be improved even further using the constraint equation~\eqref{eqn_constraintgeneralone}, as~$\Sigma_+^2\tilde A=\O(e^{(12+12\slimit-3\eps)\tau})$.
\end{rema}
\pagebreak
In the proof, we make use of the following estimate for~$\Delta$:
\begin{lemm}
\label{lemm_estimatedeltaconstraint}
	Consider a point~$(\Sigma_+,\tilde\Sigma,\Delta,\tilde A,N_+)\in\RR^5$ satisfying the constraint equation~\eqref{eqn_constraintgeneralone}, with~$\tilde N$ defined as in~\eqref{eqn_definitiontilden}. Then
	\begin{equation}
		0\le \absval\Delta \le \sqrt{\frac{\tilde\Sigma}3}\absval{N_+}+\sqrt{\frac{\absval\bparamk}{3}}\sqrt{\tilde\Sigma \tilde A}.
	\end{equation}
\end{lemm}
\begin{proof}
	From the constraint equation~\eqref{eqn_constraintgeneralone}, it follows that
	\begin{equation}
		0\le \Delta^2\le\tilde\Sigma\tilde N\le\frac13\tilde\Sigma N_+^2+\frac{\absval{\bparamk}}3\tilde\Sigma\tilde A.
	\end{equation}
	Taking the square root concludes the proof.
\end{proof}
\begin{proof}[Proof of Lemma~\ref{lemm_decaystiff_axis}]
	We start by proving the decay for~$N_+$, and to do so we distinguish between the two cases~$\tilde A>0$ and~$\tilde A=0$. We begin with the latter. In this case, we find from the definition of~$q$, equation~\eqref{eqn_qwithaandn}, and the constraint equation~\eqref{eqn_constraintgeneralone} that
	\begin{align}
		q={}&2-2\tilde N=2-\frac23 N_+^2,\\
		\Delta^2={}&\frac13\tilde\Sigma N_+^2.
	\end{align}
	Inserting this in the evolution equation~\eqref{eqns_evolutionbianchib} for~$N_+$ reveals
	\begin{align}
		N_+'={}&(2-\frac23 N_+^2+2\Sigma_+)N_+ + 2\sqrt{3\tilde\Sigma}\absval{N_+}\\
		={}&(2-\frac23N_+^2+2\Sigma_+\pm 2\sqrt{3\tilde\Sigma})N_+,
	\end{align}
	where the~$\pm$ accounts for the different possible signs of~$\Delta N_+$. 
	For either sign, the term in the bracket converges to~$2+2\slimit$, and applying Lemma~\ref{lemm_decaylemmageneral} gives the decay of~$N_+$ in the statement.
	
	If on the other hand~$\tilde A>0$, then its decay is given by Lemma~\ref{lemm_decaytildea}. We set 
	\begin{equation}
		X\coloneqq\frac{N_+}{\tilde A^{1/2}}
	\end{equation}
	and compute from the evolutions equations~\eqref{eqns_evolutionbianchib} that
	\begin{equation}
		X'=\frac{6\Delta}{\tilde A^{1/2}}
	\end{equation}
	which, by Lemma~\ref{lemm_estimatedeltaconstraint}, we can estimate with
	\begin{equation}
		\absval{X'}\le 6\left[ \sqrt{\frac{\tilde\Sigma}{3}}\frac{\absval{N_+}}{\tilde A^{1/2}} +\sqrt{\frac{\absval{\bparamk}}3}\sqrt{\tilde\Sigma}\right]
		=2\sqrt{3\tilde\Sigma}\left[\absval{X} +\sqrt{\absval{\bparamk}}\right].
	\end{equation}
	Consequently, setting
	\begin{equation}
		Y\coloneqq X^2+1,
	\end{equation}
	we find that
	\begin{equation}
		\absval{Y'}\le C \sqrt{\tilde\Sigma}\absval Y,
	\end{equation}
	where~$C>0$ is a constant. As~$\tilde\Sigma$ converges to~zero by assumption, this implies that
	\begin{equation}
		Y\le C e^{{-}\eps\tau}
	\end{equation}
	for sufficiently negative times~$\tau$. Applying the decay of~$\tilde A$, we conclude that~$N_+=\O(e^{(2+2\slimit-\eps)\tau})$.

	Having found the decay for~$N_+$ to hold independently of whether~$\tilde A$ vanishes or not, we continue with the remaining variables. We treat both cases simultaneously.
	The decay of~$\tilde N$ follows from its definition, equation~\eqref{eqn_definitiontilden}, together with the decay of~$\tilde A$ from Lemma~\ref{lemm_decaytildea}. 
	The decay for~$q$ is a direct consequence of equation~\eqref{eqn_qwithaandn}. 
	Knowing that~$\Sigma_+$ and~$\tilde\Sigma$ are at least bounded, see Remark~\ref{rema_statespacecompact}, the constraint equation~\eqref{eqn_constraintgeneralone} then implies that~$\Delta=\O(e^{(2+2\slimit-\eps/2)\tau})$.
	
	From the decay of the evolution equations~\eqref{eqns_evolutionbianchib} for~$\Sigma_+$ and~$\tilde\Sigma$,
	\begin{align}
		\Sigma_+'={}&(q-2)\Sigma_+-2\tilde N=\O(e^{(4+4\slimit-\eps)\tau}),\\
		\tilde\Sigma'={}&2(q-2)\tilde\Sigma-4\Sigma_+\tilde A-4\Delta N_+=\O(e^{(4+4\slimit-2\eps)\tau}),
	\end{align}
	combined with the assumption on the limit point, it follows that~$\Sigma_+=\slimit+\O(e^{(4+4\slimit-\eps)\tau})$ and~$\tilde\Sigma=\O(e^{(4+4\slimit-\eps)\tau})$.
	
	To conclude the proof, it remains to show that~$\slimit=0$ in case~$\tilde A>0$, and that~$\tilde\Sigma$ and~$\Delta$ have the improved decay given in the statement. We start with the case~$\tilde A>0$: 
	From the constraint equation~\eqref{eqn_constraintgeneralone}, the decay estimates we have found so far and the lower bound on the decay of~$\tilde A$ from Lemma~\ref{lemm_decaytildea}, we find that
	\begin{equation}
		0\le \Delta ^2 = \tilde\Sigma \tilde N -\Sigma_+^2\tilde A \le{-}\slimit^2 e^{(4+4\slimit+\eps)\tau} + \O(e^{(8+8\slimit-2\eps)\tau}).
	\end{equation}
	Due to the sign, this is only possible for~$\slimit=0$. From the same estimate, we conclude immediately that~$\Delta$ has the claimed decay. 
	In case~$\tilde A=0$, the constraint equation~\eqref{eqn_constraintgeneralone} reads
	\begin{equation}
		\Delta^2=\tilde\Sigma\tilde N=\O(e^{(8+8\slimit-2\eps)\tau})
	\end{equation}
	right away, which proves the decay estimates for~$\Delta$ also in this case.	
	Independently of the sign of~$\tilde A$, the evolution equation~\eqref{eqns_evolutionbianchib} for~$\tilde\Sigma$ then reads
	\begin{equation}
		\tilde\Sigma'=2(q-2)\tilde\Sigma-4\Sigma_+\tilde A-4\Delta N_+=\O(e^{(6+6\slimit-2\eps)\tau}),
	\end{equation}
	which yields the decay for~$\tilde\Sigma$.
\end{proof}
Up to now, all the decay rates we have obtained in the different settings include an~$\eps>0$. We now make use of Lemma~\ref{lemm_decaylemmageneralimproved} to obtain decay rates independent of~$\eps$. In the following lemma and proposition, statements similar to Lemma~6.13 and Proposition~6.14 in~\cite{radermacher_sccogonbianchibperfectfluidsvacuum} are achieved.
\begin{lemm}
\label{lemm_improveddecaytildeA}
	\considerconvtoD. Then either $\tilde A=0$ along the whole solution, or there are constants $c_{\tilde A},C_{\tilde A}>0$ and~$\tau_0$ such that $\tau\le\tau_0$ implies
	\begin{equation}
		c_{\tilde A} e^{(4+4\slimit)\tau}\le \tilde A(\tau)\le C_{\tilde A} e^{(4+4\slimit)\tau}.
	\end{equation}
\end{lemm}
\begin{proof}
	We wish to apply Lemma~\ref{lemm_decaylemmageneralimproved} to the evolution equation~\eqref{eqns_evolutionbianchib} for~$\tilde A$. It is enough to show that
	\begin{equation}
		2(q+\Sigma_+)
	\end{equation}
	converges exponentially to~$4+4\slimit$. 
	
	In case~$\stildelimit=0$, this is shown in Lemma~\ref{lemm_decaystiff_axis}. Let us therefore assume that~$\stildelimit>0$. We know from Lemma~\ref{lemm_deltanplussigns_stiff} that~$\Delta N_+$ has constant sign for sufficiently negative times. In Lemma~\ref{lemm_deltanplussamesign_stiff}, we cover the case~$\Delta N_+>0$. If~$\Delta=0=N_+$ along the whole solution, then Lemma~\ref{lemm_deltanpluszero_stiff} implies the result. The remaining case~$\Delta N_+<0$ is treated in Lemma~\ref{lemm_deltanplusrightoftaub2oppositesign_stiff}.
\end{proof}
With this, we now show improved convergence rates for all variables, in the following sense: So far, we have determined the limiting value and the lowest order of decay. We now determine this lowest order in more detail and even find the second lowest order. 

The statement we show applies to solutions converging to a limit point in~$\stiffalphalimitminus$ or~$\stiffalphalimitzero$ and uses the same method of proof as~\cite[Prop.~6.14]{radermacher_sccogonbianchibperfectfluidsvacuum}. The same approach could also be applied to solutions converging to a limit point in~$\stiffalphalimitplus$, however we do not make use of any such statement in this paper.
\begin{prop}
\label{prop_improvedconvergence}
	\considerconvtoD. If~$(1+\slimit)^2\le3\stildelimit$, if $\tilde A>0$, and if $\Delta$ and $N_+$ do not both vanish identically along the solution, then $3\slimit^2+\bparamk\stildelimit=0$ and there are constants $\alpha>0$, $\betaD\not=0$ and $\betaN\not=0$ satisfying
	\begin{equation}
		\betaD=\sqrt{\stildelimit/3}\betaN
	\end{equation}
	\st
	\begin{align}
	\label{eqn_leftoftaub2additionaldecay_sigmaplus_stiff}
		\Sigma_+={}&\slimit-\frac{\alpha \slimit}{2(1+\slimit)\stildelimit}(\slimit+\slimit^2+\stildelimit)e^{(4+4\slimit)\tau}+\O(e^{(8+8\slimit-\eps)\tau}),\\
	\label{eqn_leftoftaub2additionaldecay_tildesigma_stiff}
		\tilde\Sigma={}&\stildelimit
		-\frac{\alpha}{1+\slimit}(\slimit+\slimit^2+\stildelimit) e^{(4+4\slimit)\tau}+\O(e^{(8+8\slimit-\eps)\tau}),\\
		\Delta={}&\betaD e^{(2+2\slimit+2\sqrt{3\stildelimit})\tau}+\O(e^{(6+6\slimit+2\sqrt{3\stildelimit}-\eps)\tau}),\\
		\tilde A={}&\alpha e^{(4+4\slimit)\tau}+\O(e^{(8+8\slimit-\eps)\tau}),\\
		N_+={}&\betaN e^{(2+2\slimit+2\sqrt{3\stildelimit})\tau}+\O(e^{(6+6\slimit+2\sqrt{3\stildelimit}-\eps)\tau}),\\
		\tilde N={}&{-}\frac {\bparamk}3\alpha e^{(4+4\slimit)\tau}+\O(e^{(8+8\slimit-\eps)\tau}),\\
		q={}&2+2\alpha(\frac {\bparamk}3-1)e^{(4+4\slimit)\tau}+\O(e^{(8+8\slimit-\eps)\tau}),
	\end{align}
	as $\tau\rightarrow{-}\infty$, for every $\eps>0$.
\end{prop}
\begin{proof}
	As a consequence of Lemma~\ref{lemm_deltanplusoppositesign_stiff}, we know that~$\Delta N_+>0$ holds for sufficiently negative times. In this setting, the convergence rates of Lemma~\ref{lemm_deltanplussamesign_stiff} apply, and from the relation
	\begin{equation}
		\tilde A(3\slimit^2+\bparamk\stildelimit)=0
	\end{equation}
	we conclude the relation between the location of the limit point and the parameter~$\bparamk$. Note that our assumptions exclude~$\stildelimit=0$, from which we deduce that~$\bparamk\le0$.
	
	\smallskip
	
	In Lemma~\ref{lemm_improveddecaytildeA}, we have shown that~$e^{{-}(4+4\slimit)\tau}\tilde A$ is bounded for sufficiently negative~$\tau$. From the evolution equations~\eqref{eqns_evolutionbianchib}, we find that
	\begin{equation}
		(e^{-(4+4\slimit)\tau}\tilde A)'=(2(q-2)+4(\Sigma_+-\slimit))e^{-(4+4\slimit)\tau}\tilde A
	\end{equation}
	and the decay rates established in Lemma~\ref{lemm_deltanplussamesign_stiff} give $2(q-2)+4(\Sigma_+-\slimit)=\O(e^{(4+4\slimit-\eps)\tau})$. We therefore find
	\begin{equation}
	\label{eqn_auxiliarybetterdecaytildea_stiff}
		e^{-(4+4\slimit)\tau}\tilde A(\tau)=\lim_{\tau\rightarrow{-}\infty}e^{-(4+4\slimit)\tau}\tilde A(\tau)+\int_{{-}\infty}^\tau\O(e^{(4+4\slimit-\eps)\sigma})d\sigma,
	\end{equation}
	which, setting
	\begin{equation}
		\alpha\coloneqq \lim_{\tau\rightarrow{-}\infty}e^{-(4+4\slimit)\tau}\tilde A(\tau),
	\end{equation}
	is equivalent to
	\begin{equation}
		\tilde A=\alpha e^{(4+4\slimit)\tau}+\O(e^{(8+8\slimit-\eps)\tau}).
	\end{equation}
	Even though the exact value of~$\alpha$ is not known, we know that it is positive due to $\tilde A>0$, Further, this form shows that there are no terms of exponential order between $4+4\slimit$ and $8+8\slimit-\eps$.
	
	Combining the improved decay for~$\tilde A$ with the ones from Lemma~\ref{lemm_deltanplussamesign_stiff} implies that
	\begin{align}
		\tilde N={}&\frac13(N_+^2-{\bparamk}\tilde A)={-}\frac {\bparamk}3\alpha e^{(4+4\slimit)\tau}+\O(e^{(8+8\slimit-\eps)\tau}),\\
		q={}&2(1-\tilde A-\tilde N)=2+2\alpha(\frac {\bparamk}3-1)e^{(4+4\slimit)\tau}+\O(e^{(8+8\slimit-\eps)\tau}),
	\end{align}
	as $(1+\slimit)^2\le3\stildelimit$ by assumption. With this, the evolution equations~\eqref{eqns_evolutionbianchib} for~$\Sigma_+$ and~$\tilde\Sigma$ read
	\begin{align}
	\Sigma_+'={}&(q-2)\Sigma_+-2\tilde N\\
		={}&2\alpha((\frac\bparamk3-1)\slimit+\frac{\bparamk}3) e^{(4+4\slimit)\tau} +\O(e^{(8+8\slimit-\eps)\tau}),\\
	\tilde\Sigma'={}&2(q-2)\tilde\Sigma-4\Delta N_+-4\Sigma_+\tilde A\\
		={}&4\alpha((\frac\bparamk3-1)\stildelimit-\slimit )e^{(4+4\slimit)\tau} +\O(e^{(8+8\slimit-\eps)\tau}).
	\end{align}
	Using the relation between~$\slimit$, $\stildelimit$, and~$\bparamk$, integration yields the convergence rates for~$\Sigma_+$ and~$\tilde\Sigma$.
	
	\smallskip
	
	In order to determine the improved decay for~$\Delta$ and~$N_+$, we recall that in the proof of Lemma~\ref{lemm_decaydeltanplus_stiff_notaxis} we found
	\begin{equation}
	    (\Delta+\rfactorstiff N_+)'
		=(2+2\slimit+2\sqrt{3\stildelimit}+\frac{f_1\Delta+f_2N_+}{\Delta+\rfactorstiff N_+})(\Delta+\rfactorstiff N_+),
	\end{equation}
	compare equation~\eqref{eqn_evolutionrdeltaplusnplus}. Note that~$\rfactorstiff$ is defined in~\eqref{eqn_definitionrfactor_stiff} and positive. A closer look at the computation carried out there reveals that the two asymptotically vanishing functions have the form
	\begin{align}
		f_1={}&2q-4+2\Sigma_+-2\slimit,\\
		f_2={}&\rfactorstiff q-2\rfactorstiff+2\rfactorstiff\Sigma_+ -2\rfactorstiff\slimit+2\tilde\Sigma-2\stildelimit-2\tilde N.
	\end{align}
	Due to the previous results, both functions decay as~$\O(e^{(4+4\slimit-\eps)\tau})$. The quotient containing the functions~$f_1,f_2$ inherits this asymptotic decay, as~$\Delta$ and~$\rfactorstiff N_+$ have the same sign. We therefore find
	\begin{equation}
		(e^{-(2+2\slimit+2\sqrt{3\stildelimit})\tau}(\Delta+\rfactorstiff N_+))'=e^{-(2+2\slimit+2\sqrt{3\stildelimit})\tau}(\Delta+\rfactorstiff N_+)\O(e^{(4+4\slimit-\eps)\tau}),
	\end{equation}
	and from this
	\begin{equation}
		e^{-(2+2\slimit+2\sqrt{3\stildelimit})\tau}(\Delta+\rfactorstiff N_+)=\beta+\O(e^{(4+4\slimit-\eps)\tau}),
	\end{equation}
	for some constant $\beta\not=0$ having the same sign as~$\Delta$ and~$N_+$. Consequently, we obtain
	\begin{equation}
	\label{eqn_improveddecaydeltaplusrnplus}
		(\Delta+\rfactorstiff N_+)=\beta e^{(2+2\slimit+2\sqrt{3\stildelimit})\tau}+\O(e^{(6+6\slimit+2\sqrt{3\stildelimit}-\eps)\tau}).
	\end{equation}
	We conclude from Remark~\ref{rema_auxiliarydecayconstraintexpr_stiff} that
	\begin{equation}
		(\sqrt{\frac{\tilde\Sigma}{3}}N_++\Delta)(\sqrt{\frac{\tilde\Sigma}{3}}N_+-\Delta)=\frac13(3\Sigma_+^2+{\bparamk}\tilde\Sigma)\tilde A=\O(e^{(8+8\slimit+4\sqrt{3\stildelimit}-\eps)\tau}).
	\end{equation}
	Due to estimate~\eqref{eqn_improveddecaydeltaplusrnplus}, the first factor on the \lhs\ is of order $e^{(2+2\slimit+2\sqrt{3\stildelimit})\tau}$ and not faster.
	Therefore, the second factor has to compensate by decaying sufficiently fast, implying
	\begin{equation}
	\label{eqn_relationdeltanplusimproveddecay}
		\Delta=\rfactorstiff N_++\O(e^{(6+6\slimit+2\sqrt{3\stildelimit}-\eps)\tau}).
	\end{equation}
	Combining this with \equ~\eqref{eqn_improveddecaydeltaplusrnplus} yields the decay expressions for~$\Delta$ and~$N_+$ in the statement, with $\betaD=\rfactorstiff\betaN$.
	This concludes the proof.
\end{proof}

\section{Asymptotics towards the upper part of the Jacobs set~\texorpdfstring{$\stiffalphalimit$}{J}}
\label{section_convergenceunstablepart}

In this section, we discuss in more detail the behaviour of solutions converging towards the subset~$\stiffalphalimitminus\subset\stiffalphalimit$. From Proposition~\ref{prop_convergencemain_unstablepart}, we know that for such solutions either~$\tilde A=0$ has to hold, or the limit point~$(\slimit,\stildelimit,0,0,0)$ and the parameter~$\bparamk$ are related via~$3\slimit^2+\bparamk\stildelimit=0$. Note that this is only possible for non-positive values of~$\bparamk$.
This means that a solution converging to the subset~$\stiffalphalimitminus$ is either of Bianchi class~A, or it is of Bianchi type~IV, V or~VI$_{\binvparam}$. Further, the limit point lies on a specific arc in~$\stiffalphalimitminus$.

Here, we investigate the possible Bianchi class~B solutions in more detail and show that they have to be contained in sub\mf s of positive codimension. To do so, we apply concepts and results from the theory of dynamical systems, the same way we did in~\cite{radermacher_sccogonbianchibperfectfluidsvacuum} in Sections~8 and~10. For an overview over the notation and technique, we refer to~\cite[Appendix~B]{radermacher_sccogonbianchibperfectfluidsvacuum}. 
The statement we make use of is~\cite[Thm~B.3]{radermacher_sccogonbianchibperfectfluidsvacuum}, determining properties of the so-called centre-unstable \mf~$\C^u$ which contains the maximal negatively invariant set~$A^-(U)$ of some neighborhood of points in~$\stiffalphalimitminus$. This set~$A^-(U)$ is the set of all points which do not leave~$U$ under the evolution in negative time direction, and in particular contains all solutions converging to a limit point in~$U$.

We start this section with a discussion of the linearised evolution equation in the extended five-dimensional state space, by which we mean the linear approximation of the evolution equations~\eqref{eqns_evolutionbianchib}--\eqref{eqn_definitiontilden}, but without restricting them to the constraint equations~\eqref{eqn_constraintgeneralone}--\eqref{eqn_constraintgeneraltwo}.
We have given this vector field explicitly, for points contained in the Jacobs set~$\stiffalphalimit$, in Appendix~\ref{appendix_linearisedevolution}, and also stated its eigenvalues. For points in~$\stiffalphalimitminus$, exactly four out of the five eigenvalues are non-negative, counting zero twice due to being a double eigenvalue.

According to~\cite[Thm~B.3]{radermacher_sccogonbianchibperfectfluidsvacuum}, the centre-unstable \mf\ of the linearised evolution is tangent to the subspace spanned by the eigenvectors to non-negative eigenvalues. Therefore, we find that solutions converging to a limit point in~$\stiffalphalimitminus$ as~$\tau\rightarrow{-}\infty$ are contained in a certain four-dimensional \mf\ in the extended state space~$\RR^5$ for sufficiently negative times.

Compared to the asymptotic behaviour in non-stiff Bianchi fluid which has been discussed in~\cite{radermacher_sccogonbianchibperfectfluidsvacuum}, we find the following connection: In non-stiff fluids, part of the possible limit points are contained in the equilibrium arc~$(\slimit,1-\slimit^2,0,0,0)$, $\slimit\in[{-}1,1]$, called the Kasner parabola~$\kasnerparabola$, which lies in the boundary of the Jacobs set~$\stiffalphalimit$. On the arc of~$\kasnerparabola$ where~$\slimit\in({-}1,1/2)$, one out of the five eigenvalues of the linearised evolution equation is negative, the number of zero eigenvalues equals the dimension of the equilibrium set (dimension one), all others are positive. This is the same setting we encounter for the set~$\stiffalphalimitminus$ in stiff fluids. The investigation of the Kasner subarc in question, carried out in~\cite[Sect.~8]{radermacher_sccogonbianchibperfectfluidsvacuum}, reveals that possible Bianchi class~B non-stiff fluid solutions converging to this arc are contained in a  countable union of~$C^1$ sub\mf s of positive codimension. The statement we show in this section mirrors that result.
\begin{prop}
\label{prop_centreunstablemf_stiff}
	Consider the evolution equations~\eqref{eqns_evolutionbianchib}--\eqref{eqn_definitiontilden} in the extended state space, \ie without assuming the constraint equations~\eqref{eqn_constraintgeneralone}--\eqref{eqn_constraintgeneraltwo}. If~$K$ is a compact subset of~$\stiffalphalimitminus$, then there is a neighborhood~$U$ of~$K$ and a four-dimensional~$C^1$ sub\mf~$\centreunstable$ with the following properties:
	\begin{itemize}
		\item $\centreunstable$ contains the set~$K$ and in each point in~$K$ is tangent to those eigenvectors of the linearised evolution equations in the extended five-dimensional space, given in Appendix~\ref{appendix_linearisedevolution}, which correspond to eigenvalues with non-negative real parts.
		\item Points in~$U$ are either contained in~$\centreunstable$ or their evolution under equations~\eqref{eqns_evolutionbianchib}--\eqref{eqn_definitiontilden} leaves~$U$ as~$\tau\rightarrow{-}\infty$.
	\end{itemize}
\end{prop}
\begin{rema}
	We see in the proof that we could change the regularity of the manifold~$\centreunstable$ to~$C^r$, for some~$r<\infty$. Further, as all eigenvalues on~$\stiffalphalimitminus$ are real, the \mf~$\centreunstable$ is tangent to those eigenvectors which correspond to non-negative eigenvalues.
\end{rema}
\begin{rema}
	This statement in particular applies to every solution to the evolution \equ s~\eqref{eqns_evolutionbianchib}--\eqref{eqn_evolutionomega} which converges to a point $(\slimit,\stildelimit,0,0,0)\in K$ as $\tau\rightarrow{-}\infty$, for a given compact set~$K\subset\stiffalphalimitminus$: There is a time~$\tau_0$ such that the solution is contained in the neighborhood~$U$ for all times~$\tau\le\tau_0$. Therefore, the solution has to be contained in the sub\mf~$\centreunstable$ for~$\tau\le\tau_0$.
\end{rema}
\begin{proof}
	The set~$\stiffalphalimitminus$ is a two-dimensional manifold consisting of equilibrium points of the evolution equations~\eqref{eqns_evolutionbianchib}--\eqref{eqn_definitiontilden}.
	We apply~\cite[Thm~B.3]{radermacher_sccogonbianchibperfectfluidsvacuum} to find a centre-unstable \mf~$\C^u$ near~$K\subset\stiffalphalimitminus$ as well as a neighborhood~$U$ of~$K$ such that every negatively invariant set~$A^-(U)$ is contained in~$\C^u$. Without loss of generality, we restrict the \mf~$\C^u$ to the open set~$U$.
	
	In every point~$m=(\Sigma_+,\tilde\Sigma,0,0,0)\in K$, the \mf~$\C^u$ is tangent to the invariant subspaces~$E_m^s\oplus E_m^c$, see Definition~B.2 in~\cite{radermacher_sccogonbianchibperfectfluidsvacuum} and the adjacent text. In our case, these spaces are spanned by the eigenvectors to eigenvalues
	\begin{equation}
		0 \qquad 2(1+\Sigma_++\sqrt{3\tilde\Sigma}) \qquad 4(1+\Sigma_+) \qquad 0.
	\end{equation}
	Note that the eigenspace associated with the eigenvalue zero is two-dimensional, see Appendix~\ref{appendix_linearisedevolution}.
	
	The evolution equations~\eqref{eqns_evolutionbianchib} are polynomial and therefore~$C^\infty$.
	When applying~\cite[Thm~B.3]{radermacher_sccogonbianchibperfectfluidsvacuum}, we can therefore choose every finite~$r$, for example~$r=1$.
	This concludes the proof.
\end{proof}

In the previous statement, we considered the evolution in the extended state space and found that solutions converging to a point in~$\stiffalphalimitminus$ have to be contained in the \mf~$\centreunstable$ which has positive codimension in~$\RR^5$. Now, we wish to transfer the result to the set of points such that the constraint equations~\eqref{eqn_constraintgeneralone}--\eqref{eqn_constraintgeneraltwo} are satisfied. In particular, we are interested in the intersection between this constraint surface and the manifold~$\centreunstable$ we determined in Proposition~\ref{prop_centreunstablemf_stiff}, as we want to find the codimension of this intersection set in the set determined by the constraint equations.

For points on the arc in~$\stiffalphalimitminus$ we are interested in, the constraint surface is singular, see Remark~\ref{rema_singularconstraintequ}. Information on the normal and tangent directions of the constraint surface and the \mf~$\centreunstable$ in these points therefore does not suffice to conclude statements on the codimension. We additionally make use of the convergence rates we determined in Proposition~\ref{prop_improvedconvergence}.

\begin{theo}
\label{theo_unstablepartfamilymanifolds}
	Let~$\bparamk\not=0$ and consider the set of solutions to equations~\eqref{eqns_evolutionbianchib}--\eqref{eqn_evolutionomega} converging to~$(\slimit,\stildelimit,0,0,0)\in\stiffalphalimitminus$ as~$\tau\rightarrow{-}\infty$. If~$\tilde A>0$, then~$\bparamk<0$ and
	\begin{equation}
		3\slimit^2+\bparamk\stildelimit=0,
	\end{equation}
	and the solution is contained in one of the following subsets:
	\begin{itemize}
		\item The invariant set satisfying~$\tilde A>0$, $\Delta=0=N_+$, $3\Sigma_+^2+\bparamk\tilde\Sigma=0$.
		\item A countable union of $C^1$ sub\mf s satisfying $\tilde A>0$, and~$\Delta$, $N_+$ not both vanishing identically. The sub\mf s are of dimension at most three, and in the set describing Bianchi type~VI$_{\binvparam}$ solutions to the evolution equations~\eqref{eqns_evolutionbianchib}--\eqref{eqn_evolutionomega}, \ie in the set
		\begin{equation}
			B(VI_{\binvparam}) ={}\{\eqref{eqn_constraintgeneralone}-\eqref{eqn_constraintgeneraltwo}\text{ hold},\,{\bparamk}=\nicefrac1{\binvparam}<0,\, \tilde A>0\},
		\end{equation}
		have codimension at least~one.
	\end{itemize}
\end{theo}
\begin{proof}
	We have discussed convergence towards~$\stiffalphalimitminus$ in Proposition~\ref{lemm_deltanplussamesign_stiff}. The relation between~$\slimit$, $\stildelimit$ and~$\bparamk$ follows directly from this statement, as we assume~$\tilde A>0$. The case~$\bparamk=0$ is excluded by assumption, and for~$\bparamk>0$, the only way to satisfy this relation is with~$\slimit=0=\stildelimit$. However, the point~$(0,0,0,0,0)$ is not contained in~$\stiffalphalimitminus$.
	
	The set~$\Delta=0=N_+$ is invariant under the evolution equations~\eqref{eqns_evolutionbianchib}--\eqref{eqn_evolutionomega}. Every solution satisfying this relation at one time does so at all times. This is the setting of the first case, and the relation between~$\Sigma_+$ and~$\tilde\Sigma$ follows directly from the constraint equation~\eqref{eqn_constraintgeneralone} and~$\tilde A>0$.
	
	Let us turn to the second case. We fix a countable family of compact subarcs~$K_m$, $m\in\NN$, exhausting the arc~$\{(\slimit,\stildelimit,0,0,0)\in\stiffalphalimitminus\, | \,3\slimit^2+\bparamk\stildelimit=0\} $. 
	For a given solution converging to a limit point on this arc, we fix a compact set~$K_m$ containing the limit point.
	To this~$K_m$, we then apply the previous proposition and find that the solution is contained in the four-dimensional \mf~$\centreunstable$ for sufficiently negative times. At the same time, the solution has to satisfy the constraint equation~\eqref{eqn_constraintgeneralone}. 
	To conclude, we now have to show that the intersection between the constraint surface and the sub\mf~$\centreunstable$ is either empty or tranverse, if we additionally assume that~$\tilde A>0$ and~$\Delta$, $N_+$ do not both vanish identically. We do this by comparing the spanning directions of~$\centreunstable$ to the normal direction of the constraint surface, \ie the gradient of the constraint equation~\eqref{eqn_constraintgeneralone}.
	
	We consider solutions with~$\tilde A>0$, $\Delta$ and~$N_+$ not both vanishing identically, and~$\bparamk<0$. Due to Remark~\ref{rema_singularconstraintequ}, the gradient of the constraint equation~\eqref{eqn_gradientconstraintequ} is non-zero along such solutions and therefore the set of points defined by the constraint equation~\eqref{eqn_constraintgeneralone} is smooth close to such a solution. In the limit point however, the relation between~$\slimit$, $\stildelimit$ and~$\bparamk$ implies that the gradient vanishes. 
	In this point, we therefore cannot conclude transversality from computing the scalar product between the (vanishing) gradient and the vectors spanning the sub\mf~$M$. Instead, we apply the improved decay properties determined in Lemma~\ref{lemm_improveddecaytildeA} to the gradient of the constraint equation~\eqref{eqn_gradientconstraintequ} and find that close to a limit point~$(\slimit,\stildelimit,0,0,0)\in\stiffalphalimitminus$ with~$3\slimit^2+\bparamk\stildelimit=0$, the gradient has the form
	\begin{equation}
		({-}2\slimit ,{-}\frac\bparamk3 ,0,0,0)\alpha e^{(4+4\slimit)\tau} + \O(e^{(\min(8+8\slimit,2+2\slimit+2\sqrt{3\stildelimit})-\eps)\tau}).
	\end{equation}
	The last term in this expression is understood as a vector in~$\RR^5$ such that every single component has the decay given, which is faster than the decay of the first vector. We can therefore normalise the gradient of the constraint equation by multiplying it with~$e^{{-}(4+4\slimit)\tau}/\alpha$ and obtain a well-defined non-vanishing gradient direction even up to the limit point~$(\slimit,\stildelimit,0,0,0)$. Here, we make use of the fact that~$\bparamk\not=0$ and thus~$\slimit\not=0$.
	
	Consider now a small neighborhood of the subarc~$K_m$, and intersect it with the set defined by the constraints~\eqref{eqn_constraintgeneralone}--\eqref{eqn_constraintgeneraltwo}. By the previous discussion, the normal direction to the constraint surface close to~$K_m$ is a small perturbation of the normalised gradient direction we determined above, thus contained in a small cone around the vector~$({-}2\slimit ,{-}\bparamk/3 ,0,0,0)$. We are interested in points where there is a non-empty intersection with the \mf~$M$. The vectors spanning~$M$ in a point close to~$K_m$ are small perturbations of the vectors spanning~$M$ in points contained in~$K_m$. 
	By the previous proposition, these are the eigenvectors to non-negative eigenvalues of the linearised evolution equations, see Appendix~\ref{appendix_linearisedevolution} for their explicit form.
	Computing the  scalar product of these eigenvectors with the vector~$({-}2\slimit ,{-}\bparamk/3 ,0,0,0)$ reveals that the gradient direction is \ogon\ to the eigenvectors to eigenvalues~$2+2\slimit\pm2\sqrt{3\stildelimit}$ and~$4+4\slimit$. No transversality can be concluded for slightly perturbed direction vectors. For the eigenspace associated with the eigenvalue~$0$ however, we see that any vector contained in a sufficiently small cone around the vector~$({-}2\slimit ,{-}\bparamk/3 ,0,0,0)$ is contained in suitable slight perturbation of the space spanned by the vectors~$(1,0,0,0,0)$ and~$(0,1,0,0,0)$.
	We therefore conclude the following:
	In a sufficiently small neighborhood of the arc~$K_m$, the constraint surface and the \mf~$M$ intersect transversally in a sub\mf\ of dimension at most three. As the set~$B(VI_\binvparam)$ is of dimension four, this implies a positive codimension.
	
	We now apply the flow corresponding to the evolution equations and integer times to the intersection \mf\ obtained for~$K_m$. As the flow is a diffeomorphism coming from a polynomial evolution equation, the resulting set is a countable union of $C^1$ sub\mf s of dimension at most three and codimension at least~one. Taking the union over all subarcs~$K_m$ then yields a countable union of sub\mf s with the requested properties which contains all solutions listed in the second case of the statement. This concludes the proof.
\end{proof}
\begin{rema}
	In the statement, we excluded the case~$\bparamk=0$. It becomes obvious in the proof why this is necessary when using the method we chose: As a consequence of the relation \begin{equation}
		3\slimit^2+\bparamk\stildelimit=0,
	\end{equation}
	allowing for~$\bparamk=0$ would imply that~$\slimit=0$ is no longer excluded. In this case the rescaled gradient direction
	\begin{equation}
		({-}2\slimit ,{-}\frac\bparamk3 ,0,0,0),
	\end{equation}
	while non-vanishing for all solution where~$\bparamk\not=0$ and~$\slimit\not=0$, becomes the null vector, for which computing the scalar product to check \ogon ity becomes meaningless.
	
	For a solution of Bianchi types~IV or~V, which implies~$\bparamk=0$, we however do know that the limit point has to satisfy~$\slimit=0$.
\end{rema}

\section{Asymptotics towards the special arc in the Jacobs set~\texorpdfstring{$\stiffalphalimit$}{J}}
\label{section_convergencespecialarc}

In this section, we investigate the behaviour of solutions converging towards a limit point on the arc~$\stiffalphalimitzero\subset\stiffalphalimit$. 
Making use of Proposition~\ref{prop_convergencemain_unstablepart}, we find that Bianchi type~VII$_{\binvparam}$ solutions are excluded altogether, and in the other Bianchi types of class~B, solutions converge to one of at most two specific points on~$\stiffalphalimitzero$, determined by the Bianchi type and the parameter~$\bparamk$.

On the arc~$\stiffalphalimitzero$, not two but three of the eigenvalues to the linearised evolution equations turn zero, see Appendix~\ref{appendix_linearisedevolution}. This is one more than the dimension of the Jacobs set~$\stiffalphalimit$ and two more than the dimension of the arc itself, and prohibits us from applying dynamical system theory the way we did in Section~\ref{section_convergenceunstablepart}.

In the setting of non-stiff Bianchi perfect fluids investigated in~\cite{radermacher_sccogonbianchibperfectfluidsvacuum}, a similar situation occurs for the point Taub~2: There as well, one of the eigenvalues to the linearised evolution equations on the Kasner parabola~$\kasnerparabola$ can have either sign, and it vanishes exactly in the point Taub~2. Due to the lower dimension in the non-stiff setting, one point on the Kasner parabola~$\kasnerparabola$ versus an arc in the Jacobs set~$\stiffalphalimit$, a direct method could be applied to determine the set containing all solutions converging to a limit point with additional zero eigenvalue. Such an approach has not been found for the stiff fluid setting and the arc~$\stiffalphalimitzero$. Nonetheless, we collect the statements we achieve without such a method.
\begin{prop}
\label{prop_specialarc}
	Consider the set of solutions to equations~\eqref{eqns_evolutionbianchib}--\eqref{eqn_evolutionomega} converging to~$(\slimit,\stildelimit,0,0,0)\in\stiffalphalimitzero$ as~$\tau\rightarrow{-}\infty$. If~$\tilde A>0$, then~$\bparamk\le0$ and the limit point has to satisfy
	\begin{equation}
		3\slimit^2+\bparamk\stildelimit=0.
	\end{equation}
\end{prop}
\begin{proof}
	For a solution of class~B converging to a point in~$\stiffalphalimitzero$, the relation between~$\bparamk$ and the location of the limit point has been proven in Proposition~\ref{prop_convergencemain_unstablepart}. If~$\bparamk>0$, this relation holds only for~$\slimit=0=\stildelimit$, but this point is not contained in~$\stiffalphalimitzero$. Therefore, in this case no solution converges to~$\stiffalphalimitzero$. 
\end{proof}

\begin{rema}
\label{rema_atmosttwopointsstiffzero}
	Consider solutions to equations~\eqref{eqns_evolutionbianchib}--\eqref{eqn_evolutionomega} and assume that~$\tilde A>0$, \ie restrict to Bianchi class~B. We have found in the previous proposition that all solutions converging to the arc~$\stiffalphalimitzero$ as~$\tau\rightarrow{-}\infty$ have a limit point contained in~$\stiffalphalimitzero\cap\{3\Sigma_+^2+\bparamk\tilde\Sigma=0\}$. We discuss this intersection set depending on the different Bianchi types:	
	\begin{itemize}
		\item Bianchi type~IV and~V: As~$\bparamk=0$, the limit point is the intersection between the arc~$\stiffalphalimitzero$ and the~$\tilde\Sigma$-axis, hence the point
		\begin{equation}
			(0,1/3,0,0,0).
		\end{equation}
		\item Bianchi type~VI$_{\binvparam}$: As~$\bparamk<0$, the set containing all possible limit points is the intersection set of the two parabolas~$\stiffalphalimitzero$ and~$\{3\Sigma_+^2+\bparamk\tilde\Sigma=0\}$ contained in the~$\Sigma_+\tilde\Sigma$-plane. The following has to hold for points satisfying both equations: If~$\bparamk\not={-}9$, then
		\begin{equation}
		\label{eqn_intersectionpoint}
			\Sigma_+=\frac{\pm\sqrt{{-}\bparamk}}{3\mp\sqrt{-\bparamk}},\qquad\tilde\Sigma=\frac3{(3\mp\sqrt{-\bparamk})^2},
		\end{equation}
		otherwise
		\begin{equation}
			(\Sigma_+,\tilde\Sigma)=({-}\frac12,\frac1{12}).
		\end{equation}
		Note that in case~$\bparamk={-}9$, we have excluded 'exceptional' spacetimes, see Remark~\ref{rema_excludeexceptionalBianchi}.
		
		For~$\bparamk={-}1$, the intersection point~\eqref{eqn_intersectionpoint} with the upper choice of sign is the point Taub~2. Using a straight-forward geometrical argument, we therefore conclude that the set~$\stiffalphalimitzero\cap\{3\Sigma_+^2+\bparamk\tilde\Sigma=0\}$ consists of one or two points, depending on the value of~$\bparamk$:
		\begin{itemize}
			\item For~$\bparamk\in({-}\infty,{-}1]$, the upper choice of sign in~\eqref{eqn_intersectionpoint} yields an intersection point located in the complement of~$\stiffalphalimit$. The intersection point with lower choice of sign is located on~$\stiffalphalimitzero$ and has a negative~$\Sigma_+$-value.
			\item For~$\bparamk\in({-}1,0)$, both sign choices in~\eqref{eqn_intersectionpoint} yield intersection points located on~$\stiffalphalimitzero$. The upper choice of sign yields a positive, the lower choice a negative~$\Sigma_+$-value.
		\end{itemize}
	\item Bianchi type~VII$_{\binvparam}$: As~$\bparamk>0$, convergence to~$\stiffalphalimitzero$ is excluded by Proposition~\ref{prop_specialarc}.
	\end{itemize}
\end{rema}

\section{Proof of the main theorem}
\label{section_proofmaintheorem}

We conclude the paper with a proof of the main theorem stated in the introduction, Theorem~\ref{theo_fullmeasurelimitset}. We do so by combining the findings from the previous sections where we have discussed in detail the decay behaviour of the individual variables as well as the asymptotic properties of solutions upon convergence to the different subsets of the Jacobs set~$\stiffalphalimit$.
\begin{proof}[Proof of Theorem~\ref{theo_fullmeasurelimitset}]
	We have seen in Proposition~\ref{prop_convergencetoD} that every solution to the evolution equations~\eqref{eqns_evolutionbianchib}--\eqref{eqn_evolutionomega} converges to a limit point contained in the Jacobs set~$\stiffalphalimit$ as~$\tau\rightarrow{-}\infty$. We prove the theorem by showing that all solutions which have a convergence behaviour different from the one given in the statement are either impossible or have to be contained in countable unions of smooth sub\mf s of positive codimension.
	
	For solutions with~$\tilde A>0$ converging to a limit point~$(\slimit,\stildelimit,0,0,0)\in\stiffalphalimitminus$, Proposition~\ref{prop_convergencemain_unstablepart} implies that~$3\slimit^2+\bparamk\stildelimit=0$ has to hold. In the case~$\bparamk=0$ this implies that~$\slimit=0$.
	The remaining case~$\bparamk\not=0$ has been investigated in Theorem~\ref{theo_unstablepartfamilymanifolds}: Solutions have to be of Bianchi type~VI$_{\binvparam}$, and are contained in a countable union of~$C^1$ sub\mf s of positive codimension.
	
	In Proposition~\ref{prop_specialarc} and Remark~\ref{rema_atmosttwopointsstiffzero}, we have discussed convergence to limit points in~$\stiffalphalimitzero$ and found that Bianchi type~VII$_{\binvparam}$ solutions are excluded. Bianchi type~IV and~V solutions converge to the limit point~$(0,1/3,0,0,0)$. For Bianchi type~VI$_{\binvparam}$ solutions, the limit point has to coincide with one of the at most two elements of the intersection~$\stiffalphalimitzero\cap\{3\Sigma_+^2+\bparamk\tilde\Sigma=0\}$.

	To conclude the proof, we collect the results in the different Bianchi types. 
	For Bianchi type~VII$_{\binvparam}$ solutions, convergence to~$\stiffalphalimitminus$ and~$\stiffalphalimitzero$ is excluded, all solutions must converge to~$\stiffalphalimitplus$.
	In Bianchi type~VI$_{\binvparam}$, solutions which do not converge to~$\stiffalphalimitplus$ can converge to one of the points on~$\stiffalphalimitzero$ determined above, or they converge to~$\stiffalphalimitminus$ and consequently are contained in the countable union of~$C^1$ sub\mf s with positive codimension given in Theorem~\ref{theo_unstablepartfamilymanifolds}. The complement of this union is of full measure and a countable intersection of open and dense sets.
	In Bianchi type~IV, solutions not converging to~$\stiffalphalimitplus$ have a limit point in~$\stiffalphalimitminus$ or~$\stiffalphalimitzero$ satisfying~$\Sigma_+=0$.
	In the case of Bianchi~V, $\Sigma_+=0$ holds along the solution and thus up to the limit point. As all limit points of solutions to the evolution equations~\eqref{eqns_evolutionbianchib}--\eqref{eqn_evolutionomega} are contained in~$\stiffalphalimit$, this concludes the proof.	
\end{proof}

\appendix

\section{The linearised evolution equations}
\label{appendix_linearisedevolution}

In this section, we discuss the linear approximation of the evolution equations~\eqref{eqns_evolutionbianchib}, with~$q$ and~$\tilde N$ defined as in equations~\eqref{eqn_qwithaandn} and~\eqref{eqn_definitiontilden}. Throughout the paper, we refer to this approximation as the linearised evolution equations in the extended five-dimensional state space, because we neglect the constraint equations~\eqref{eqn_constraintgeneralone}--\eqref{eqn_constraintgeneraltwo}.
For points~$(\Sigma_+,\tilde\Sigma,0,0,0)$ contained in the Jacobs set~$\stiffalphalimit$, this is the linear map~$\RR^5\rightarrow\RR^5$ given by the matrix
\begin{equation}
	\begin{pmatrix}
	0 & 0 & 0 & 2(\frac{\bparamk}3-1)\Sigma_++\frac23\bparamk & 0\\
	0 & 0 & 0 & 4(\frac{\bparamk}3-1)\tilde\Sigma-4\Sigma_+ & 0 \\
	0 & 0 & 2\Sigma_++2 & 0 & 2\tilde\Sigma \\
	0 & 0 & 0 & 4+4\Sigma_+ & 0 \\
	0 & 0 & 6 & 0 & 2+2\Sigma_+
	\end{pmatrix}.
\end{equation}
Note that this coincides with the linearised evolution of a non-stiff fluid in Kasner points, given in~\cite[App~A.1]{radermacher_sccogonbianchibperfectfluidsvacuum}, when setting~$\gamma=2$.

As there appear to be typos in the eigenvalues in~\cite[Sect.~4.4]{hewittwainwright_dynamicalsystemsapproachbianchiorthogonalB},
we give here the corrected eigenvalues and state the corresponding eigenvectors.\begin{itemize}
	\item The eigenvalue~$0$ is a double eigenvalue, with a two-dimensional eigenspace tangential to the Jacobs set~$\stiffalphalimit$ spanned by
	\begin{equation}
		(1\,,\,0\,,\,0\,,\,0\,,\,0) \qquad \text{and} \qquad (0\,,\,1\,,\,0\,,\,0\,,\,0).
	\end{equation}
	\item The eigenspace to eigenvalue $2(1+\Sigma_++\sqrt{3\tilde\Sigma})$ lies in the $\Delta N_+$-plane and is spanned by
	\begin{equation}
		(0\,,\,0\,,\,+\frac13\sqrt{3\tilde\Sigma}\,,\,0\,,\,1).
	\end{equation}
	\item The eigenspace to eigenvalue $2(1+\Sigma_+-\sqrt{3\tilde\Sigma})$ lies in the $\Delta N_+$-plane and is spanned by
	\begin{equation}
		(0\,,\,0\,,\,{-}\frac13\sqrt{3\tilde\Sigma}\,,\,0\,,\,1).
	\end{equation}
	\item The eigenspace to eigenvalue $4(1+\Sigma_+)$ is spanned by
	\begin{equation}
		((\frac\bparamk6-\frac12)\Sigma_++\frac{\bparamk}{6}\,,\,{-}\Sigma_++(\frac\bparamk3-1)\tilde\Sigma \,,\,0\,,\,\Sigma_++1\,,\,0).
	\end{equation}
\end{itemize}

\section{The evolution equations for perfect fluid \ogon\ Bianchi class~B models}
\label{appendix_nonstiffevolution}

The setting which we discuss in this paper is a special case of a perfect fluid \ogon\ Bianchi class~B solution. Expansion-normalised variables for the general setting have been introduced in~\cite{hewittwainwright_dynamicalsystemsapproachbianchiorthogonalB} and their asymptotic behaviour towards the initial singularity is discussed in~\cite{radermacher_sccogonbianchibperfectfluidsvacuum}. In addition to the Bianchi parameter~$\bparamk$ which determines the Bianchi type, a parameter~$\gamma\in[0,2]$ is given which specifies the perfect fluid. Compared to the evolution equations~\eqref{eqns_evolutionbianchib}--\eqref{eqn_evolutionomega} which we have given in the introduction, the following two relations differ in the general perfect fluid case, being the only ones which include this parameter~$\gamma$: The decelaration parameter satisfies
\begin{equation}
\label{eqn_definitionqgeneral}
	q=\frac32(2-\gamma)(\Sigma_+^2+\tilde\Sigma)+\frac12(3\gamma-2)(1-\tilde A-\tilde N),
\end{equation}
and the density parameter~$\Omega$ evolves according to
\begin{equation}
\label{eqn_evolutionomegageneral}
	\Omega'=(2q-(3\gamma-2))\Omega.
\end{equation}
Restricting to the case of a stiff fluid is achieved by setting~$\gamma$ to the extremal value~$2$ and assuming~$\Omega>0$ in the general set of evolution equations, see for example~\cite[(5)--(11)]{radermacher_sccogonbianchibperfectfluidsvacuum}.

\bibliographystyle{amsalpha}
\bibliography{researchbib}
\vfill

\end{document}

%% file: header_paperenglish.tex
\usepackage[T1]{fontenc}
\usepackage[utf8]{inputenc}
\usepackage{lmodern}
\usepackage{amsmath,amssymb}
\usepackage{mathrsfs}
\usepackage{amsfonts,amsthm,latexsym}
\usepackage[arrow, matrix, curve]{xy}
\usepackage{mathtools}
\usepackage{nicefrac}
\usepackage{tikz}
\usepackage{bracket-hack}
\usepackage{ming}
\usepackage[english]{babel}
\usepackage{graphicx}
\usepackage{color}
\usepackage{paralist}
\usepackage[final]{hyperref}
\usepackage[colorinlistoftodos,obeyDraft]{todonotes}
\usepackage[notcite,notref]{showkeys}
\usepackage{marginnote}
\usepackage{ifdraft}
\usepackage[a4paper,twoside]{geometry}

\mathtoolsset{showonlyrefs=true}
\mathtoolsset{centercolon}
\setdefaultenum{i)}{}{}{}
\usepackage{auto-eqlabel}

\theoremstyle{definition}
\newtheorem{defi}{Definition}[section]

\newtheorem{rema}[defi]{Remark}

\theoremstyle{plain}
\newtheorem{prop}[defi]{Proposition}
\newtheorem{theo}[defi]{Theorem}

\newtheorem{lemm}[defi]{Lemma}

\newenvironment{subaligned}{\left\{\aligned}{\endaligned\right.}

\newcommand{\R}{{\mathcal R}}

\newcommand{\D}{{\mathcal D}}

\newcommand{\K}{{\mathcal K}}

\renewcommand{\O}{{\mathcal O}}
\newcommand{\C}{{\mathcal C}}

\renewcommand{\S}{{\mathcal S}}

\newcommand{\J}{{\mathcal J}}

\newcommand{\maths}[1]{{\ensuremath{\mathbb #1}}}
\newcommand{\RR}{\maths{R}}
\newcommand{\NN}{\maths{N}}

\newcommand{\ie}{i.\,e.\ }
\newcommand{\st}{such that }

\newcommand{\rhs}{right-hand side}
\newcommand{\lhs}{left-hand side}

\newcommand{\equ}{equation}
\newcommand{\equs}{equations}

\newcommand{\mf}{mani\-fold}

\newcommand{\riem}{Rie\-mann\-ian}

\newcommand{\liegr}{Lie group}

\newcommand{\fundform}{second fun\-da\-men\-tal form}
\newcommand{\curv}{curvature}

\newcommand{\ogon}{or\-tho\-go\-nal}

\newcommand{\mghd}{ma\-xi\-mal glo\-bally hy\-per\-bo\-lic de\-ve\-lop\-ment}
\newcommand{\desitter}{de~Sitter}

\newcommand{\eps}{\varepsilon}

\newcommand{\absval}[1]{\lvert #1 \rvert}

\newcommand{\kasnerparabola}{\K}

%% file: Radermacher_Bianchistifffluid.bbl
\newcommand{\etalchar}[1]{$^{#1}$}
\providecommand{\bysame}{\leavevmode\hbox to3em{\hrulefill}\thinspace}
\providecommand{\MR}{\relax\ifhmode\unskip\space\fi MR }
\providecommand{\MRhref}[2]{%
  \href{http://www.ams.org/mathscinet-getitem?mr=#1}{#2}
}
\providecommand{\href}[2]{#2}
\begin{thebibliography}{KBS{\etalchar{+}}03}

\bibitem[HW93]{hewittwainwright_dynamicalsystemsapproachbianchiorthogonalB}
C.~G. Hewitt and John Wainwright, \emph{A dynamical systems approach to
  {B}ianchi cosmologies: orthogonal models of class {B}}, Classical Quantum
  Gravity \textbf{10} (1993), no.~1, 99--124. \MR{1200591 (94b:83027)}

\bibitem[KBS{\etalchar{+}}03]{krasinskibehrschueckingestabrookwahlquistellisjantzenkundt_bianchiclass}
A.~Krasi{\'n}ski, Christoph~G. Behr, Engelbert Sch{\"u}cking, Frank~B.
  Estabrook, Hugo~D. Wahlquist, George F.~R. Ellis, Robert Jantzen, and
  Wolfgang Kundt, \emph{The {B}ianchi classification in the
  {S}ch\"ucking-{B}ehr approach}, Gen. Relativity Gravitation \textbf{35}
  (2003), no.~3, 475--489. \MR{1964375}

\bibitem[Rad16]{radermacher_sccogonbianchibperfectfluidsvacuum}
Katharina Radermacher, \emph{Strong {C}osmic {C}ensorship in orthogonal
  {B}ianchi class~{B} perfect fluids and vacuum models},
  \href{https://arxiv.org/abs/1612.06278v3}{arXiv:1612.06278v3 [gr-qc]}.

\bibitem[WH89]{wainwrighthsu_dynamicalsystemsapproachbianchiorthogonalA}
John Wainwright and Lucas Hsu, \emph{A dynamical systems approach to {B}ianchi
  cosmologies: orthogonal models of class {$A$}}, Classical Quantum Gravity
  \textbf{6} (1989), no.~10, 1409--1431. \MR{1014971 (90h:83033)}

\end{thebibliography}
